\documentclass[draftclsnofoot,onecolumn,10pt]{IEEEtran}

\usepackage{amsmath,amssymb,enumerate,subfigure,multirow,hhline,color}
\usepackage{xcolor}
\usepackage{hyperref}
\usepackage{graphicx}
\usepackage{graphics}
\usepackage[sort]{cite}
\usepackage{amsthm}
\usepackage{algorithm}
\usepackage{algpseudocode}

\usepackage{tcolorbox}
\usepackage{epsfig,psfrag}
\usepackage{tabularx}
\usepackage{tabulary}

\usepackage[sort]{cite}
\usepackage{hyperref}
\hypersetup{breaklinks=true,colorlinks=true,linkcolor=blue,citecolor=blue}

\usepackage[noabbrev,capitalise,nameinlink]{cleveref}

\DeclareMathOperator*{\argmin}{arg\,min}
\DeclareMathOperator*{\argmax}{arg\,max}

\newtheorem{theorem}{Theorem}

\newtheorem{assumption}{Assumption}

\newtheorem{corollary}{Corollary}

\newtheorem{lemma}{Lemma}
\newtheorem{definition}{Definition}
\newtheorem{proposition}{Proposition}

\allowdisplaybreaks

\title{Finite-Time Analysis of Minimax Q-Learning for Two-Player Zero-Sum Markov Games: Switching System Approach}

\author{Donghwan Lee
\thanks{D. Lee is with the Department of Electrical Engineering,
KAIST, Daejeon, 34141, South Korea {\tt\small
donghwan@kaist.ac.kr}.}
}

\begin{document}

\maketitle

\begin{abstract}
The objective of this paper is to investigate the finite-time analysis of a Q-learning algorithm applied to two-player zero-sum Markov games. Specifically, we establish a finite-time analysis of both the minimax Q-learning algorithm and the corresponding value iteration method. To enhance the analysis of both value iteration and Q-learning, we employ the switching system model of minimax Q-learning and the associated value iteration. This approach provides further insights into minimax Q-learning and facilitates a more straightforward and insightful convergence analysis. We anticipate that the introduction of these additional insights has the potential to uncover novel connections and foster collaboration between concepts in the fields of control theory and reinforcement learning communities.

\end{abstract}

\section{Introduction}

Reinforcement learning (RL) addresses the problem of optimal sequential decision-making for unknown Markov decision processes through experiences~\cite{sutton1998reinforcement}. Recent successes of RL algorithms surpassing human performance in various challenging tasks have sparked a surge of interest in both the theoretical and experimental aspects of RL algorithms~\cite{mnih2015human,wang2016dueling,lillicrap2016continuous,heess2015memory,van2016deep,bellemare2017distributional,schulman2015trust,schulman2017proximal}. Among many others, Q-learning~\cite{watkins1992q} stands out as one of the most fundamental and popular RL algorithms, with extensive studies conducted on its convergence over the past decades. Classical analysis primarily focuses on asymptotic convergence~\cite{tsitsiklis1994asynchronous,jaakkola1994convergence,borkar2000ode,hasselt2010double,melo2008analysis,lee2020unified,devraj2017zap}. However, recent advancements have been made in finite-time convergence analysis~\cite{szepesvari1998asymptotic,kearns1999finite,even2003learning,azar2011speedy,beck2012error,wainwright2019stochastic,qu2020finite,li2020sample,chen2021lyapunov,lee2020periodic,lee2021discrete}, which quantifies the speed at which iterations progress towards the solution. Most existing results consider Q-learning dynamics as nonlinear stochastic approximations~\cite{kushner2003stochastic} and utilize the contraction property of the Bellman equation. Recently,~\cite{lee2020unified,lee2021discrete} proposed a novel perspective on Q-learning based on continuous-time or discrete-time switching system models~\cite{liberzon2003switching,lin2009stability} and established asymptotic or finite-time analysis using tools from control theory~\cite{chen1995linear,khalil2002nonlinear}. This switching system perspective captures unique characteristics of Q-learning dynamics and enables the conversion of finite-time convergence analysis into stability analysis of dynamic control systems, which will also play an important role in developing main results in this paper.

In this paper, the main goal is to study Q-learning algorithms for the two-player zero-sum Markov game~\cite{shapley1953stochastic}, which is a more general Marokv decision process~\cite{puterman2014markov}, where two decision-making agents coexist and compete with each other. There exist two categories of the two-player Markov games, the alternating two-player Markov game and simultaneous two-player Markov game. In the alternating two-player Markov game, two agents engaged in decision making take turns in selecting actions to maximize and minimize cumulative discounted rewards (referred to as ``return''), respectively. On the other hand, in the simultaneous two-player Markov game, the two agents take actions simultaneously to maximize and minimize the return. Hereafter, these two agents will be called the user and the adversary. The user's primary goal entails maximizing the return, while the adversary strives to hinder the user's progress by minimizing the return. Specifically, in the alternating Markov game, the user initiates its decision at each time step without knowledge of the adversary's action. Afterwards, the adversary can observe the user's action and subsequently take its action based on the observed user's action. Consequently, the adversary holds more advantages over the user in the alternating Markov game. On the other hand, the two agents have fair chances in the simultaneous Markov game. The objective of the Markov game is to determine the pair $(\pi^*,\mu^*)$ of user's optimal policy, $\pi^*$, and adversary's optimal policy $\mu^*$. It is worth noting that although two-player Markov games represent a relatively restricted category within the realm of multi-agent environments, they possess individual significance. Moreover, the two-player Markov game includes Markov decision processes as a special case and serve as an initial stage for delving into the study of more general multi-agent Markov games.
The main objective of this paper is to establish a finite-time analysis of the minimax Q-learning method introduced in~\cite{littman1994markov} for solving two-player zero-sum Markov games. A comprehensive finite-time error analysis of minimax Q-learning as well as the corresponding value iteration is established. To facilitate the analysis of both the value iteration and Q-learning, we employ the switching system models proposed in~\cite{lee2021discrete}.

\paragraph{Contributions}
The main contributions of this paper can be summarized as follows:
\begin{enumerate}
\item This paper presents a finite-time analysis of minimax Q-learning, which has not been explored in the existing literature to the best of the author's knowledge. It is worth noting that the available existing results in the literature primarily focus on aspects such as asymptotic convergence~\cite{littman2001value,fan2019theoretical,zhu2020online} or convergence of modified algorithms~\cite{diddigi2022generalized}.

\item This paper reveals new prospects for the recently developed switching system framework~\cite{lee2021discrete}, which provides additional insights into minimax Q-learning. It also provides a conceptually simpler and more insightful convergence analysis of minimax Q-learning. We expect that the introduction of this additional insight holds the potential to unveil new connections and foster synergy among notions in control and RL communities. Furthermore, it can present additional opportunities for the development and analysis of new or other RL algorithms.
\end{enumerate}

It is important to emphasize that although the switching system model presented in~\cite{lee2021discrete} has been utilized as a foundational tool, the main analysis and specific proof techniques employed in this work substantially differ from those in~\cite{lee2021discrete}. In order to establish our central proof, we have encountered challenges that are far from trivial to overcome. To be more precise, the switching system model presented in~\cite{lee2021discrete} permits a linear comparison system that serves as a lower bound for the original system, playing a pivotal role in the finite-time analysis. Conversely, the switching system model employed in this study exhibits a distinct structure, utilizing the max-min operator in lieu of the standard max operator found in traditional Q-learning. Consequently, the switching system under consideration in this work does not admit a linear comparison system. Hence, the application of similar techniques as those employed in~\cite{lee2021discrete} is not feasible.

Moreover, it is worth noting that this paper only addresses an i.i.d. observation model, where constant step-sizes are employed to simplify the overall analysis. The i.i.d. observation model is commonly utilized in the existing literature, serving as a standard setting~\cite{dalal2018finite,borkar2000ode,melo2008analysis,lee2020unified}. Moreover, the proposed analysis can be extended to encompass the more intricate Markovian observation model, utilizing techniques presented in previous works such as~\cite{srikant2019,bhandari2018finite}. However, extending the analysis to the Markovian observation scenarios can considerably complicate the main analysis and potentially obscure the fundamental ideas and insights of our proposed approach. For this reason, in the interest of maintaining clarity and coherence, this paper will not cover the Markovian observation scenarios.

\paragraph{Related works}

The seminal work by Littman~\cite{littman1994markov} introduced minimax Q-learning, a Q-learning algorithm designed for zero-sum two-player Markov games, which serves as the main focus of our study. Subsequently, Littman and Szepesvari~\cite{littman1996generalized} established the asymptotic convergence of minimax Q-learning towards the optimal value derived from game theory.
Hu and Wellman~\cite{hu1998multiagent} extended minimax Q-learning to multi-agent environments, presenting Nash Q-learning, a variant that addresses general-sum games by incorporating Nash equilibrium computation within the learning rule. Bowling~\cite{bowling2000convergence} elucidated the convergence conditions of the algorithm, while Hu and Wellman~\cite{hu2003nash} examined its convergence behavior and highlighted the restrictive nature of the convergence assumptions. Littman et. al.~\cite{littman2001friend} introduced friend-or-foe Q-learning for general-sum Markov games, demonstrating stronger convergence properties compared to Nash Q-learning. Moreover, Littman~\cite{littman2001value} further examined convergence of Nash Q-learning and its behavior under different environments. Lagoudakis and Parr~\cite{lagoudakis2002value} studied a value iteration version of minimax Q-learning and proposed a least-squares policy iteration algorithm to solve two-player Markov games.

In recent research, Diddigi et al.~\cite{diddigi2022generalized} presented a novel generalized minimax Q-learning algorithm and provided a proof of its asymptotic convergence utilizing stochastic approximation techniques under the assumption of iterates' boundedness. Fan et al.~\cite{fan2019theoretical} extended the minimax Q-learning algorithm by incorporating deep Q-learning techniques~\cite{mnih2015human} and established a finite-time error bound. Zhu and Zhao~\cite{zhu2020online} also employed deep Q-learning techniques in the context of minimax Q-learning and demonstrated its asymptotic convergence in tabular learning scenarios.

Additionally, several notable studies on Markov games, while not directly focused on minimax Q-learning, offer valuable insights. Srinivasan et al.~\cite{srinivasan2018actor} and Perolat et al.~\cite{perolat2018actor} investigated actor-critic algorithms tailored for multi-agent Markov games. Perolat et al.~\cite{perolat2015approximate} delved into an approximate dynamic programming framework for two-player zero-sum Markov games. Furthermore, Perolat et al.~\cite{perolat2016use} explored the generalization of various non-stationary RL algorithms and provide theoretical analyses. Wei et al.~\cite{wei2017online} examined online reinforcement learning algorithms designed for average-reward two-player Markov games.
Lastly, Zhang et al.~\cite{zhang2021multi} presented a comprehensive survey on the multi-agent Markov game and multi-agent reinforcement learning.

It is important to acknowledge that although significant progress has been made through these prior works over the years, the existing literature primarily focuses on aspects such as asymptotic convergence~\cite{littman2001value, fan2019theoretical, zhu2020online} or the convergence of modified algorithms~\cite{diddigi2022generalized}. To the authors' best knowledge, a rigorous finite-time convergence analysis of minimax Q-learning has yet to be thoroughly investigated.

\section{Preliminaries and problem formulation}

\subsection{Notation}
The adopted notation is as follows: ${\mathbb R}$: set of real numbers; ${\mathbb R}^n $: $n$-dimensional Euclidean
space; ${\mathbb R}^{n \times m}$: set of all $n \times m$ real
matrices; $A^T$: transpose of matrix $A$; $A \succ 0$ ($A \prec
0$, $A\succeq 0$, and $A\preceq 0$, respectively): symmetric
positive definite (negative definite, positive semi-definite, and
negative semi-definite, respectively) matrix $A$; $I$: identity matrix with appropriate dimensions; $\lambda_{\min}(A)$ and $\lambda_{\max}(A)$ for any symmetric matrix $A$: the minimum and maximum eigenvalues of $A$; $|{\cal S}|$: cardinality of a finite set $\cal S$; ${\rm tr}(A)$: trace of any matrix $A$; $A \otimes B$: Kronecker’s product of matrices $A$ and $B$.

\subsection{Markov decision problem}
For reference, we first briefly introduce the standard Markov decision problem (MDP)~\cite{puterman2014markov,bertsekas1996neuro}, where a decision making agent sequentially takes actions to maximize cumulative discounted rewards in environments called Markov decision process. A Markov decision process is a mathematical model of dynamical systems with the state-space ${\cal S}:=\{ 1,2,\ldots ,|{\cal S}|\}$ and action-space ${\cal A}:= \{1,2,\ldots,|{\cal A}|\}$. The decision maker selects an action $a \in {\cal A}$ with the current state $s$, then the state
transits to a state $s'$ with probability $P(s,a,s')$, and the transition incurs a
reward $r(s,a,s')$, where $P(s,a,s')$ is the state transition probability from the current state
$s\in {\cal S}$ to the next state $s' \in {\cal S}$ under action $a \in {\cal A}$, and $r(s,a,s')$ is the reward function. For convenience, we consider a deterministic reward function and simply write $r(s_k,a_k ,s_{k + 1}) =:r_k,k \in \{ 0,1,\ldots \}$.

A deterministic policy, $\pi :{\cal S} \to {\cal A}$, maps a state $s \in {\cal S}$ to an action $\pi(s)\in {\cal A}$. The objective of the Markov decision problem (MDP) is to find a deterministic optimal policy, $\pi^*$, such that the cumulative discounted rewards over infinite time horizons is
maximized, i.e.,
\begin{align*}
\pi^*:= \argmax_{\pi\in \Theta} {\mathbb E}\left[\left.\sum_{k=0}^\infty {\gamma^k r_k}\right|\pi\right],
\end{align*}
where $\gamma \in [0,1)$ is the discount factor, $\Theta$ is the set of all admissible deterministic policies, $(s_0,a_0,s_1,a_1,\ldots)$ is a state-action trajectory generated by the Markov chain under policy $\pi$, and ${\mathbb E}[\cdot|\pi]$ is an expectation conditioned on the policy $\pi$. The Q-function under policy $\pi$ is defined as
\begin{align*}
&Q^{\pi}(s,a)={\mathbb E}\left[ \left. \sum_{k=0}^\infty {\gamma^k r_k} \right|s_0=s,a_0=a,\pi \right],\quad s\in {\cal S},a\in {\cal A},
\end{align*}
and the optimal Q-function is defined as $Q^*(s,a)=Q^{\pi^*}(s,a)$ for all $s\in {\cal S},a\in {\cal A}$. Once $Q^*$ is known, then an optimal policy can be retrieved by the greedy policy $\pi^*(s)=\argmax_{a\in {\cal A}}Q^*(s,a)$.

\subsection{Two-player zero-sum Markov game}
In this paper, we consider a two-player zero-sum Markov game, where two decision making agents sequentially take actions to maximize and minimize cumulative discounted rewards (return), respectively. Hereafter, these two agents will be called the user and the adversary. The user's primary goal entails maximizing the return, while the adversary strives to hinder the user's progress by minimizing the return. There exist two categories of the two-player Markov games, the alternating two-player Markov game and simultaneous two-player Markov game. In the alternating two-player Markov game, two agents engaged in decision making take turns in selecting actions to maximize and minimize the return, respectively. On the other hand, in the simultaneous two-player Markov game, the two agents take actions simultaneously to maximize and minimize the return. Specifically, in the alternating Markov game, the user initiates its decision at each time step without knowledge of the adversary's action. Afterwards, the adversary can observe the user's action and subsequently take its action based on the observed user's action. Consequently, the adversary holds more advantages over the user in the alternating Markov game. On the other hand, the two agents have fair chances in the simultaneous Markov game.

In this paper, we mainly focus on the alternating two-player Markov game because it simplifies the overall concepts and derivation processes, and all the results in this paper can be easily extended to the simultaneous Markov games.
In the alternating Markov game, the state-space is ${\cal S}:=\{ 1,2,\ldots ,|{\cal S}|\}$, the action-space of the user is ${\cal A}:= \{1,2,\ldots,|{\cal A}|\}$, and the action-space of the adversary is ${\cal B}:= \{1,2,\ldots,|{\cal B}|\}$. The user selects an action $a \in {\cal A}$ at the current state $s$, and the adversary can observe the user's action $a$, and selects an adversarial decision $b \in {\cal B}$. Then, the state transits to the next state $s'$ with probability $P(s'|s,a,b)$, and the transition incurs a reward $r(s,a,b,s')$, where $P(s'|s,a,b)$ is the state transition probability from the current state
$s\in {\cal S}$ to the next state $s' \in {\cal S}$ under actions $a \in {\cal A}$, $b \in {\cal B}$, and $r(s,a,b,s')$ is the reward function. For convenience, we consider a deterministic reward function and simply write $r(s_k,a_k,b_k ,s_{k + 1}) =:r_k,k \in \{ 0,1,\ldots \}$. The user's stationary deterministic policy, $\pi :{\cal S} \to {\cal A}$, maps a state $s \in {\cal S}$ to an action $\pi(s)\in {\cal A}$. The adversary's stationary deterministic policy, $\mu :{\cal S} \times {\cal A} \to {\cal B}$, maps a state $s \in {\cal S}$ and the user's action $a \in {\cal A}$ to an adversarial action $\mu(s,a)\in {\cal B}$. The user does not have an access to the adversary's action, while the adversary can observe the user's action before making its decision.
It is known that there exists an optimal stationary deterministic policy~\cite{littman1994markov} for both the user and adversary. The objective of the Markov game is to determine the user's optimal policy, denoted as $\pi^*$, and the optimal adversarial policy, denoted as $\mu^*$:
\[
(\pi ^*,\mu^*) : = \argmax _{\pi  \in \Theta } \min _{\mu  \in \Omega } {\mathbb E}\left[ {\left. {\sum\limits_{k = 0}^\infty  {\gamma ^k r_k } } \right|\pi ,\mu } \right],
\]
where $\gamma \in [0,1)$ is the discount factor, $\Theta$ is the set of all admissible deterministic policies of the user, $\Omega$ is the set of all admissible deterministic policies of the adversary, $(s_0,a_0,b_0,s_1,a_1,b_1,\ldots)$ is a state-action trajectory generated under policies $\pi$, $\mu$, and ${\mathbb E}[\cdot|\pi,\mu]$ is an expectation conditioned on the policies $\pi$ and $\mu$.

The Markov game considered in this paper can be potentially applied to the following scenarios:
\begin{enumerate}
\item Adversarial decision making: There exists an intelligent adversary that can take adversarial actions to prevent the user from achieving their goal. Under this situation, the user wants to find an optimal policy that can achieve the best possible performance against the adversarial behaviors.

\item Robust decision making: The environment changes arbitrarily, and the user want to find a robust optimal policy that can achieve the best possible performance in the worst case scenarios.
\end{enumerate}

Some fundamental tools and notions used in Markov decision problem~\cite{bertsekas2015dynamic}, such as the value function and Bellman equation, can be also applied to the two-player Markov game.
In particular, the optimal Q-function is defined as
\[
Q^* (s,a,b): = \max _{\pi  \in \Theta } \min _{\mu  \in \Omega } {\mathbb E}\left[ {\left. {\sum\limits_{k = 0}^\infty  {\gamma ^k r_k } } \right|s_0  = s,a_0  = a,b_0  = b,\pi ,\mu } \right]
\]
which satisfies the optimal Q-Bellman equation
\begin{align}
Q^* (s,a,b) = \underbrace {R(s,a,b) + \gamma \sum\limits_{s' \in S} {P(s'|s,a,b)\max _{a' \in {\cal A}} \min _{b' \in {\cal B}} Q^* (s',a',b')} }_{: = (FQ^* )(s,a,b)}\label{eq:optimal-Q-Bellman-eq}
\end{align}

The corresponding user's stationary optimal policy is given by
\[
\pi ^* (s) = \arg \max _{a \in {\cal A}} \min _{b \in {\cal B}} Q^* (s,a,b)
\]
and, the adversary's stationary optimal policy is
\[
\mu ^* (s,a) = \arg \min _{b \in {\cal B}} Q^* (s,a,b)
\]

Using the optimal Bellman equation~\eqref{eq:optimal-Q-Bellman-eq}, the Q-value iteration (Q-VI) is the recursion
\begin{align}
Q_{k + 1}(s,a,b)  = (FQ_k)(s,a,b),\quad (s,a,b) \in {\cal S} \times {\cal A} \times {\cal B}
\end{align}
with any $Q_0\in {\mathbb R}^{|{\cal S} \times {\cal A} \times {\cal B}|}$. It is known that the Q-VI converges to $Q^*$~\cite{lagoudakis2002value,zhu2020online}.

\subsection{Switching system}

In this paper, the proposed analysis mainly relies on the so-called switching system models~\cite{liberzon2003switching,lin2009stability} in the control community. Therefore, we briefly introduce the notion of switching systems here. Since the switching system is a special form of nonlinear systems~\cite{khalil2002nonlinear}, we first consider the nonlinear system
\begin{align}
x_{k+1}=f(x_k),\quad x_0=z \in {\mathbb R}^n,\quad k\in \{1,2,\ldots \},\label{eq:nonlinear-system}
\end{align}
where $x_k\in {\mathbb R}^n$ represents the state and $f:{\mathbb R}^n \to {\mathbb R}^n$ denotes a nonlinear mapping. An essential concept when dealing with the nonlinear system is the equilibrium point. A point $x^*\in {\mathbb R}^n$ in the state-space is said to be an equilibrium point of~\eqref{eq:nonlinear-system} if it has the property that when the system's state begins at $x^*$, it remains at $x^*$~\cite{khalil2002nonlinear}. For\eqref{eq:nonlinear-system}, the equilibrium points are the real roots of the equation $f(x) = x$. The equilibrium point $x^*$ is said to be globally asymptotically stable if, for any initial state $x_0 \in {\mathbb R}^n$, $x_k \to x^*$ as $k \to \infty$.

Next, let us consider the particular system, called the \emph{linear switching system},
\begin{align}
&x_{k+1}=A_{\sigma_k} x_k,\quad x_0=z\in {\mathbb
R}^n,\quad k\in \{0,1,\ldots \},\label{eq:switched-system}
\end{align}
where $x_k \in {\mathbb R}^n$ is the state, $\sigma\in {\mathcal M}:=\{1,2,\ldots,M\}$ is called the mode, $\sigma_k \in
{\mathcal M}$ is called the switching signal, and $\{A_\sigma,\sigma\in {\mathcal M}\}$ are called the subsystem matrices. The switching signal can be either arbitrary or controlled by the user under a certain switching policy. Especially, a state-feedback switching policy is denoted by $\sigma_k = \sigma(x_k)$. A more general class of systems is the {\em affine switching system}
\begin{align*}
&x_{k+1}=A_{\sigma_k} x_k + b_{\sigma_k},\quad x_0=z\in {\mathbb
R}^n,\quad k\in \{0,1,\ldots \},
\end{align*}
where $b_{\sigma_k} \in {\mathbb R}^n$ is the additional input vector, which also switches according to $\sigma_k$. Due to the additional input $b_{\sigma_k}$, its stabilization becomes much more challenging.

\section{Convergence of Q-VI via switching system model}

In this section, we provide a proof of convergence of Q-VI drawn from switching system models. This approach yields additional perspectives on Q-VI, supplementing the existing analysis in the literature. Moreover, it will serve as a fundamental basis for the convergence analysis of minimax Q-learning in the subsequent sections. To represent Q-VI compactly through the switching system model, we need to introduce some vector and matrix notations.

\subsection{Assumptions and definitions}\label{section:notations}
Throughout the paper, we will use the following compact notations for dynamical system representations of Q-VI:
\begin{align}
R_{a,b} : = \left[ {\begin{array}{*{20}c}
   {R(1,a,b)}  \\
   {R(2,a,b)}  \\
    \vdots   \\
   {R(|{\cal S}|,a,b)}  \\
\end{array}} \right] \in {\mathbb R}^{|{\cal S}|} ,\quad P_{a,b} : = \left[ {\begin{array}{*{20}c}
   {P(1|1,a,b)} & {P(2|1,a,b)} &  \cdots  & {P(|{\cal S}||1,a,b)}  \\
   {P(1|2,a,b)} & {P(2|2,a,b)} &  \cdots  & {P(|{\cal S}||2,a,b)}  \\
    \vdots  &  \vdots  &  \ddots  &  \vdots   \\
   {P(1||{\cal S}|,a,b)} & {P(2||{\cal S}|,a,b)} &  \cdots  & {P(|{\cal S}|||{\cal S}|,a,b)}  \\
\end{array}} \right]\label{eq:R-P}
\end{align}
where $R_{a,b}\in {\mathbb R}^{|{\cal S}|}$ is the expected reward vector conditioned on the action pair $(a,b)\in {\cal A} \times {\cal B}$, and $P_{a,b}\in {\mathbb R}^{|{\cal S}|\times |{\cal S}|}$ is the state transition probability matrix conditioned on the action pair $(a,b)\in {\cal A}\times {\cal B}$. Moreover, let us define the associated notations
\[
P: = \left[ {\begin{array}{*{20}c}
   {P_{1,1} }  \\
    \vdots   \\
   {P_{|{\cal A}|,|{\cal B}|} }  \\
\end{array}} \right] \in {\mathbb R}^{|{\cal S}| \times |{\cal S}\times {\cal A}\times {\cal B}|} ,\quad R: = \left[ {\begin{array}{*{20}c}
   {R_{1,1} }  \\
    \vdots   \\
   {R_{|{\cal A}||{\cal B}|} }  \\
\end{array}} \right] \in {\mathbb R}^{|{\cal S}\times {\cal A}\times {\cal B}|} ,\quad Q: = \left[ {\begin{array}{*{20}c}
   {Q_{1,1} }  \\
    \vdots   \\
   {Q_{|{\cal A}|,|{\cal B}|} }  \\
\end{array}} \right] \in {\mathbb R}^{|{\cal S}\times {\cal A}\times {\cal B}|} ,
\]
where $Q_{a,b}\in {\mathbb R}^{|{\cal S}|}$ is a vector with $[Q_{a,b} ]_s  = Q(s,a,b)$ and $P\in{\mathbb R}^{|{\cal S}\times {\cal A}| \times |{\cal S}|  }$. The Q-function is encoded as a single vector $Q \in {\mathbb R}^{|{\cal S}\times {\cal A}\times {\cal B}|}$, which enumerates $Q(s,a,b)$ for all $s \in {\cal S}$, $a \in {\cal A}$, and $b \in {\cal B}$. The single value $Q(s,a,b)$ can be extracted by
\[
Q(s,a,b) = (e_a  \otimes e_b  \otimes e_s )^T Q,
\]
where $e_s \in {\mathbb R}^{|{\cal S}|}$, $e_a \in {\mathbb R}^{|{\cal A}|}$, and $e_b \in {\mathbb R}^{|{\cal B}|}$ are $s$-th basis vector (all components are $0$ except for the $s$-th component which is $1$), $a$-th basis vector, and $b$-th basis vector, respectively.

For any given $Q \in {\mathbb R}^{|{\cal S} \times {\cal A}\times {\cal B}|}$, let us denote the greedy policy with respect to $Q$ by $i(s,a):=\argmin_{b\in {\cal B}} Q(s,a,b)\in {\cal B}$.
Then, we define the corresponding action transition matrix as
\[
\Gamma_Q  : = \left[ {\begin{array}{*{20}c}
   {e_{i(1,1)}^T  \otimes e_1^T }  \\
   {e_{i(1,2)}^T  \otimes e_2^T }  \\
    \vdots   \\
   {e_{i(|{\cal S}|,|{\cal A}|)}^T  \otimes e_{|{\cal S}||{\cal A}|}^T }  \\
\end{array}} \right] \in {\mathbb R}^{|{\cal S} \times {\cal A}| \times |{\cal S}\times {\cal A} \times {\cal B}|}
\]
where $e_{i(s,a)}  \in {\mathbb R}^{|{\cal B}|}$, and $e_i \in {\mathbb R}^{|{\cal A} \times {\cal S}|}$.
This notation has been introduced in~\cite{wang2007dual}, and it is useful to express Q-VI in the vector and matrix form using the relation
\[
\Gamma_Q  Q = \left[ {\begin{array}{*{20}c}
   {\min _{b \in {\cal B}} Q(1,1,b)}  \\
   {\min _{b \in {\cal B}} Q(1,2,b)}  \\
    \vdots   \\
   {\min _{b \in {\cal B}} Q(|{\cal S}|,|{\cal A}|,b)}  \\
\end{array}} \right] \in {\mathbb R}^{|{\cal S}\times {\cal A}|}
\]

Similarly, for any given $Q' \in {\mathbb R}^{|{\cal S}\times {\cal A}|}$, let us denote the greedy policy with respect to $Q'$ by $i(s): = \arg \max _{a \in A} Q'(s,a) \in A$.
Then, we define the corresponding action transition matrix as
\[
\Pi_Q : = \left[ {\begin{array}{*{20}c}
   { e_{i(1)}^T  \otimes e_1^T }  \\
   {e_{i(2)}^T  \otimes e_2^T }  \\
    \vdots   \\
   {e_{i(|{\cal S}|)}^T  \otimes e_{|{\cal S}|}^T }  \\
\end{array}} \right] \in {\mathbb R}^{|{\cal S}| \times |{\cal S}||{\cal A}|}
\]
where $e_{i(s)}  \in {\mathbb R}^{|{\cal A}|}$ and $e_s \in {\mathbb R}^{|{\cal S}|}$. Then, we can similarly prove that
\[
\Pi_{Q'} Q' = \left[ {\begin{array}{*{20}c}
   {\max _{a \in {\cal A}} Q'(1,a)}  \\
   {\max_{a \in {\cal A}} Q'(2,a)}  \\
    \vdots   \\
   {\max _{a \in {\cal A}} Q'(|{\cal S}|,a)}  \\
\end{array}} \right] \in {\mathbb R}^{|{\cal S}|}
\]

Combining the two notations, one can prove the relation
\[
\Pi _{\Gamma _Q Q} \Gamma _Q Q = \left[ {\begin{array}{*{20}c}
   {\max _{a \in {\cal A}} \min _{b \in {\cal B}} Q(1,a,b)}  \\
   {\max _{a \in {\cal A}} \min _{b \in {\cal B}} Q(2,a,b)}  \\
    \vdots   \\
   {\max _{a \in {\cal A}} \min _{b \in {\cal B}} Q(|{\cal S}|,a,b)}  \\
\end{array}} \right] \in {\mathbb R}^{|{\cal S}|}
\]

Another important property of the notations is that $P\Pi_{\Gamma_Q Q} \Gamma_Q \in {\mathbb R}^{|{\cal S}\times {\cal A}\times {\cal B}| \times |{\cal S}\times {\cal A}\times {\cal B}|}$ is the transition probability matrix of the state-action pair $(s,a,b)$ under the policy $(\pi,\mu)$ where $\pi (s): = \argmax _{a \in {\cal A}} \min_{b \in {\cal B}} Q(s,a,b) \in A$ and $\mu(s,a):=\argmin_{b\in {\cal B}} Q(s,a,b)\in {\cal B}$.

Using these notations, the Bellman equation in~\eqref{eq:optimal-Q-Bellman-eq} can be compactly written as
\begin{align}
Q^*  = \gamma P\Pi _{\Gamma _{Q^* } Q^* } \Gamma _{Q^* } Q^*  + R\label{eq:optimal-Q-Bellman-eq2}
\end{align}

In what follows, an equivalent switching system model, that captures the behavior of Q-VI, is introduced, and based on it, we provide a proof of convergence of Q-VI from the switching system perspective.

\subsection{Convergence of Q-VI via switching system model}

In this section, we study a discrete-time switching system model of Q-VI and establish its finite-time convergence based on the stability analysis of switching systems.
Using the notation introduced in~\cref{section:notations}, the update of Q-VI can be rewritten as
\begin{align}
Q_{k+1}= R+\gamma P\Pi_{\Gamma_{Q_k}Q_k}\Gamma_{Q_k}Q_k ,\label{eq:Q-VI-vector}
\end{align}

Combining~\eqref{eq:Q-VI-vector} and~\eqref{eq:optimal-Q-Bellman-eq2} leads to
\begin{align}
(Q_{k + 1}  - Q^* ) = \underbrace {\gamma P\Pi _{\Gamma _{Q_k } } \Gamma _{Q_k } }_{: = A_{Q_k } }(Q_k  - Q^* ) + \underbrace {\gamma P(\Pi _{\Gamma _{Q_k } Q_k } \Gamma _{Q_k }  - \Pi _{\Gamma _{Q^* } Q^* } \Gamma _{Q^* } )Q^* }_{: = b_{Q_k } }
\label{eq:swithcing-system-form0}
\end{align}
which is a switched affine system where $A_{Q_k}$ and $b_{Q_k}$ switch among matrices from $\{\gamma P\Pi_{\Gamma_Q} \Gamma_Q: Q \in \mathbb{R}^{|{\cal S} \times {\cal A} \times {\cal B}|}\}$ and vectors from $\{\gamma P(\Pi_{\Gamma_Q} \Gamma_Q - \Pi_{\Gamma_{Q^*}} \Gamma_{Q^*})Q^*: Q \in \mathbb{R}^{|{\cal S} \times {\cal A} \times {\cal B}|}\}$ based on the changes of $Q_k$.
Hence, the convergence of Q-VI now relies on analyzing the stability of the aforementioned switching system. The main challenge in proving its stability arises from the presence of the affine term $b_{Q_k}$. Without it, we could easily establish the exponential stability of the corresponding deterministic switching system under any switching policy.
Specifically, we have the following result.
\begin{proposition}\label{prop:stability}
For arbitrary $H_k, k\ge 0$, the linear switching system
\begin{align*}
Q_{k+1} - Q^* &= A_{H_k} (Q_k - Q^*),\quad  Q_0 - Q^*\in {\mathbb R}^{|{\cal S}\times {\cal A}\times {\cal B}|},
\end{align*}
is exponentially stable such that $\|Q_{k+1}- Q^*\|_\infty\le \gamma \|Q_k - Q^*\|_\infty,\quad k \ge 0$, and
\[
\|Q_k- Q^*\|_\infty\le \gamma ^k \|Q_0 - Q^*\|_\infty,\quad k \ge 0,
\]
\end{proposition}

The above result follows immediately from the key fact that $\|A_{Q} \|_\infty \le \gamma$, which we formally state in the lemma below.
\begin{lemma}\label{lemma:max-norm-system-matrix}
For any $Q \in {\mathbb R}^{|{\cal S}\times {\cal A}\times {\cal B}|}$, $\|A_{Q} \|_\infty \le \gamma$, where the matrix norm  $\| A \|_\infty :=\max_{1\le i \le m} \sum_{j=1}^n {|A_{ij} |}$ and $A_{ij}$ is the element of $A$ in $i$-th row and $j$-th column.
\end{lemma}
\begin{proof}
Note $\sum\limits_j | [A_Q ]_{ij} | = \gamma \sum\limits_j {|[P\Pi _{\Gamma _Q } \Gamma _Q ]_{ij} |}  = \gamma$, which comes from the fact that ${P\Pi _{\Gamma _Q } \Gamma _Q }$ is a stochastic matrix, i.e., its row vector is a stochastic vector. This completes the proof.
\end{proof}

However, due to the presence of the additional affine term $b_{Q_k}$ in the switching system~\eqref{eq:swithcing-system-form0}, it is not immediately evident how to directly obtain its finite-time convergence. To overcome the challenge posed by the affine term, we will utilize two simpler upper and lower bounds, provided below.
\begin{proposition}[Upper and lower bounds]\label{prop:lower-bound2}
For all $k\geq0$, we have
\[
\gamma P\Pi _{\Gamma _{Q_k } Q^* } \Gamma _{Q_k } (Q_k  - Q^* ) \le Q_{k + 1}  - Q^*  \le \gamma P\Pi _{\Gamma _{Q^* } Q_k } \Gamma _{Q^* } (Q_k  - Q^* )
\]
\end{proposition}
\begin{proof}
First of all, the lower bound can be derived though the inequalities
\begin{align*}
Q_{k + 1}  - Q^*  =& A_{Q_k } (Q_k  - Q^* ) + b_{Q_k }\\
=& \gamma P\Pi _{\Gamma _{Q_k } Q_k } \Gamma _{Q_k } Q_k  - \gamma P\Pi _{\Gamma _{Q^* } Q^* } \Gamma _{Q^* } Q^*\\
\ge& \gamma P\Pi _{\Gamma _{Q_k } Q_k } \Gamma _{Q_k } Q_k  - \gamma P\Pi _{\Gamma _{Q_k } Q^* } \Gamma _{Q_k } Q^*\\
\ge& \gamma P\Pi _{\Gamma _{Q_k } Q^* } \Gamma _{Q_k } Q_k  - \gamma P\Pi _{\Gamma _{Q_k } Q^* } \Gamma _{Q_k } Q^*\\
=& \gamma P\Pi _{\Gamma _{Q_k } Q^* } \Gamma _{Q_k } (Q_k  - Q^* ),
\end{align*}
where the inequalities come from the definitions of $\Gamma_Q$ and $\Pi_Q$. Similarly, for the upper bound, one gets
\begin{align*}
Q_{k + 1}  - Q^*  =& A_{Q_k } (Q_k  - Q^* ) + b_{Q_k }\\
 =& \gamma P\Pi _{\Gamma _{Q_k } Q_k } \Gamma _{Q_k } Q_k  - \gamma P\Pi _{\Gamma _{Q^* } Q^* } \Gamma _{Q^* } Q^*\\
\le& \gamma P\Pi _{\Gamma _{Q^* } Q_k } \Gamma _{Q^* } Q_k  - \gamma P\Pi _{\Gamma _{Q^* } Q^* } \Gamma _{Q^* } Q^*\\
\le& \gamma P\Pi _{\Gamma _{Q^* } Q_k } \Gamma _{Q^* } Q_k  - \gamma P\Pi _{\Gamma _{Q^* } Q_k } \Gamma _{Q^* } Q^*\\
=& \gamma P\Pi _{\Gamma _{Q^* } Q_k } \Gamma _{Q^* } (Q_k  - Q^* )
\end{align*}
This completes the proof.
\end{proof}

Based on the upper and lower bounds presented in~\cref{prop:lower-bound2}, one can now establish the convergence of Q-VI through the following lemma.
\begin{lemma}\label{thm:1}
We have the following bounds for Q-VI iterates:
\begin{align*}
\left\| {Q_{k+1}  - Q^* } \right\|_\infty   \le \gamma \left\| {Q_k  - Q^* } \right\|_\infty,\quad \forall k \geq 0.
\end{align*}
\end{lemma}
\begin{proof}
Since $\gamma P\Pi _{\Gamma _{Q_k } Q^* } \Gamma _{Q_k } (Q_k  - Q^* ) \le Q_{k + 1}  - Q^*  \le \gamma P\Pi _{\Gamma _{Q^* } Q_k } \Gamma _{Q^* } (Q_k  - Q^* )$ from~\cref{prop:lower-bound2}, it follows that $(e_a  \otimes e_b  \otimes e_s )^T \gamma P\Pi _{\Gamma _{Q_k } Q^* } \Gamma _{Q_k } (Q_k  - Q^* ) \le (e_a  \otimes e_b  \otimes e_s )^T (Q_{k + 1}  - Q^* ) \le (e_a  \otimes e_b  \otimes e_s )^T \gamma P\Pi _{\Gamma _{Q^* } Q_k } \Gamma _{Q^* } (Q_k  - Q^* )$.

If $(e_a  \otimes e_b  \otimes e_s )^T (Q_{k + 1}  - Q^* ) \le 0$, then $|(e_a  \otimes e_b  \otimes e_s )^T (Q_{k + 1}  - Q^* )| \le |(e_a  \otimes e_b  \otimes e_s )^T \gamma P\Pi _{\Gamma _{Q_k } Q^* } \Gamma _{Q_k } (Q_k  - Q^* )|$, where $e_s \in {\mathbb R}^{|{\cal S}|}$ and $e_a \in {\mathbb R}^{|{\cal A}|}$ are the $s$-th and $a$-th standard basis vectors, respectively. If $(e_a  \otimes e_b  \otimes e_s )^T (Q_{k + 1}  - Q^* ) > 0$, then $|(e_a  \otimes e_b  \otimes e_s )^T (Q_k  - Q^* )| \le |(e_a  \otimes e_b  \otimes e_s )^T \gamma P\Pi _{\Gamma _{Q^* } Q_k } \Gamma _{Q^* } (Q_k  - Q^* )|$. Therefore, one gets
\begin{align*}
&\left\| {Q_{k + 1}  - Q^* } \right\|_\infty   \le \max \left\{ {\left\| {\gamma P\Pi _{\Gamma _{Q_k } Q^* } \Gamma _{Q_k } (Q_k  - Q^* )} \right\|_\infty  ,\left\| {\gamma P\Pi _{\Gamma _{Q^* } Q_k } \Gamma _{Q^* } (Q_k  - Q^* )} \right\|_\infty  } \right\}\\
\le& \max \{ \gamma \left\| {Q_k  - Q^* } \right\|_\infty  ,\gamma \left\| {Q_k  - Q^* } \right\|_\infty  \}\\
=& \gamma \left\| {Q_k  - Q^* } \right\|_\infty  ,
\end{align*}
which completes the proof.
\end{proof}

As a direct consequence of~\cref{thm:1}, convergence of Q-VI can be derived as follows:
\begin{align}
\left\| {Q_k  - Q^* } \right\|_\infty   \le \gamma^k \left\| {Q_0  - Q^* } \right\|_\infty\label{eq:3}
\end{align}

In this section, we have presented a discrete-time switching system model of Q-VI and proved its convergence for alternating two-player zero-sum Markov games. It is important to note that all the derivations in this section can be readily extended to Q-VI for simultaneous two-player zero-sum Markov games. However, for the sake of simplicity in our presentation, we only focus on the alternating case. The presented switching system model serves as the basis for analyzing the minimax Q-learning algorithm in the following sections.

\section{Minimax Q-learning}
\begin{algorithm}[t]
\caption{Minimax Q-learning}
  \begin{algorithmic}[1]
    \State Initialize $Q_0 \in {\mathbb R}^{|{\cal S}\times {\cal A}\times {\cal B}|}$ randomly such that $\left\| {Q_0 } \right\|_\infty   \le 1$.
    \For{iteration $k=0,1,\ldots$}
    	\State Sample $a_k\sim \beta(\cdot|s_k)$. $b_k\sim \phi(\cdot|s_k)$ and $s_k\sim p(\cdot)$
        \State Sample $s_k'\sim P(\cdot|s_k,a_k,b_k)$ and $r_k= r(s_k,a_k,b_k,s_k')$
        \State Update
\[
Q_{k + 1} (s_k ,a_k ,b_k ) = Q_k (s_k ,a_k ,b_k ) + \alpha \left\{ {r_k  + \gamma \max _{a \in {\cal A}} \min _{b \in {\cal B}} Q_k (s_{k'} ,a,b) - Q_k (s_k ,a_k ,b_k )} \right\}
\]

    \EndFor
  \end{algorithmic}\label{algo:standard-Q-learning2}
\end{algorithm}

In this section, we study minimax Q-learning algorithm given in~\cref{algo:standard-Q-learning2} to solve the alternating two-player zero-sum Markov game.
\Cref{algo:standard-Q-learning2} is slightly different from the original minimax Q-learning proposed in~\cite{littman1994markov} for simultaneous two-player Markov games by replacing the max operator over the set of all stochastic policies with the max operator restricted to the discrete action set $\cal A$.
However, it is worth noting that all the analyses presented in this paper remain applicable to the original minimax Q-learning approach for simultaneous two-player Markov games, requiring only minor adjustments.

Through the minimax Q-learning, both the user and adversary can learn their optimal policies. However, we will focus on the user's role in this paper.
We will address the following scenario: while learning, the user has an access to the adversary's action $b\in \cal B$, which is generated by an exploratory behavior policy, meaning that the adversary does not intervene and disrupt the user.
On the other hand, after the learning period, the adversary intervenes and hides its decision to the user. Once $Q^*$ is found, the user takes action according to $\pi ^* (s): = \argmax _{a \in {\cal A}} \min _{b \in {\cal B}} Q^* (s,a,b)$, and it leads to the best performance against the optimal adversary behaviors.

In~\cref{algo:standard-Q-learning2}, we consider a constant step-size $\alpha \in (0,1)$, and assume that  $\{(s_k,a_k,b_k,s_k')\}_{k=0}^{\infty}$ are i.i.d. samples under the behavior policies $\beta$ and $\phi$, where the behavior policy is the policy by which the RL agent actually behaves to collect experiences. For simplicity, we assume that the state at each time is sampled from the stationary state distribution $p$, and in this case, the state-action distribution at each time is identically given by
\[
d(s,a,b) = p(s)\beta (a|s)\phi (b|s),\quad (s,a,b) \in {\cal S} \times {\cal A} \times {\cal B}.
\]

Throughout, we make the following assumptions.
\begin{assumption}\label{assumption:positive-distribution}
$d(s,a,b)> 0$ holds for all $s\in {\cal S},a \in {\cal A},b \in {\cal B}$.
\end{assumption}
\begin{assumption}\label{assumption:step-size}
The step-size is a constant $\alpha \in (0,1)$.
\end{assumption}
\begin{assumption}[Unit bound on rewards]\label{assumption:bounded-reward} The reward is bounded as follows:
\begin{align*}
\max _{(s,a,b,s') \in {\cal S} \times {\cal A}\times {\cal B}  \times {\cal S}} |r (s,a,b,s')|\leq 1.
\end{align*}
\end{assumption}
\begin{assumption}[Unit bound on initial parameters]\label{assumption:bounded-Q0} The initial iterate $Q_0$ satisfies $\left\| {Q_0 } \right\|_\infty   \le 1$.
\end{assumption}

The above assumptions are crucial for the proposed finite-time analysis. \cref{assumption:positive-distribution} ensures that every state-action pair can be visited infinitely often, facilitating sufficient exploration, which is a standard assumption in the literature~\cite{bertsekas1996neuro}.
This assumption can be used when the state-action occupation frequency is given, and has been also considered in~\cite{li2020sample} and~\cite{chen2021lyapunov}.
\cite{beck2012error} considers another exploration condition, called the cover time condition, which states that there exists a certain time period, within which all the state-action pair is expected to be visited at least once.
Slightly different cover time conditions have been used in~\cite{even2003learning} and~\cite{li2020sample} for convergence rate analysis.
\cref{assumption:bounded-reward} and~\cref{assumption:bounded-Q0} impose unit bounds on the reward function and the initial iterate $Q_0$, and are introduced for the sake of simplicity in analysis, without sacrificing generality.
The constant step-size in~\cref{assumption:step-size} has been also studied in~\cite{beck2012error} and~\cite{chen2021lyapunov} using different approaches.
The following quantities will be frequently used in this paper; hence, we define them for convenience.
\begin{definition}
\begin{enumerate}
\item Maximum state-action occupation frequency:
\[
d_{\max} := \max_{(s,a,b)\in {\cal S} \times {\cal A}\times {\cal B}} d(s,a,b) \in (0,1).
\]

\item Minimum state-action occupation frequency:
\[
d_{\min}:= \min_{(s,a,b) \in {\cal S} \times {\cal A} \times {\cal B}} d(s,a,b) \in (0,1).
\]

\item Exponential decay rate:
\begin{align}\label{eq:rho}
    \rho:=1 - \alpha d_{\min} (1-\gamma).
\end{align}

\end{enumerate}
\end{definition}
It can be proven that under~\cref{assumption:step-size}, the decay rate satisfies $\rho \in (0,1)$.
The reason behind referring to $\rho$ as an exponential decay rate is that the finite-time error bound, which will be derived in the remaining part of this paper, decays exponentially at the rate of $\rho$.

Similar to Q-VI, we will represent minimax Q-learning in~\cref{algo:standard-Q-learning2} as a switching system model in this section. However, a key distinction lies in the fact that~\cref{algo:standard-Q-learning2} can be viewed as a stochastic Q-VI, where the Q-function for each state-action pair is updated asynchronously through stochastic state-action pair explorations. Therefore, it becomes essential to incorporate the state-action occupation frequency, which is linked to exploration, into the switching system model. Specifically, the state-action occupation frequency is encoded using the following matrix notations:
\[
D_{a,b} : = \left[ {\begin{array}{*{20}c}
   {d(1,a,b)} & {} & {}  \\
   {} &  \ddots  & {}  \\
   {} & {} & {d(|S|,a,b)}  \\
\end{array}} \right] \in {\mathbb R}^{|{\cal S}| \times |{\cal S}|} ,\quad D: = \left[ {\begin{array}{*{20}c}
   {D_{1,1} } & {} & {}  \\
   {} &  \ddots  & {}  \\
   {} & {} & {D_{|{\cal A}|,|{\cal B}|} }  \\
\end{array}} \right] \in {\mathbb R}^{|{\cal S}\times {\cal A}\times {\cal B}| \times |{\cal S}\times {\cal A}\times {\cal B}|} .
\]
Note also that under~\cref{assumption:positive-distribution}, $D$ is a nonsingular diagonal matrix with strictly positive diagonal elements.

In our analysis, the boundedness of Q-learning iterates~\cite{gosavi2006boundedness} plays an important role in our analysis.
\begin{lemma}[Boundedness of Q-learning iterates~\cite{gosavi2006boundedness}]\label{lemma:bounded-Q}
If the step-size is less than one, then for all $k \ge 0$,
\begin{align*}
\left\| {Q_k } \right\|_\infty   \le Q_{\max } : = \frac{{\max \{ 1,\left\| {Q_0 } \right\|_\infty  \} }}{{1 - \gamma }}.
\end{align*}
\end{lemma}
From~\cref{assumption:bounded-Q0}, we can easily see that $Q_{\max}\leq\frac{1}{1-\gamma}$.
The boundedness has been established for standard Q-learning in~\cite{gosavi2006boundedness}, but not for minimax Q-learning. For this reason, we provide its proof in Appendix~\ref{appendix:1} for the completeness of the presentation.

\subsection{Minimax Q-learning as a stochastic affine switching system}
Using the notation introduced, the update in~\cref{algo:standard-Q-learning2} can be rewritten as
\begin{align}
Q_{k + 1}  = Q_k  + \alpha \{ DR + \gamma DP\Pi _{\Gamma _{Q_k } Q_k } \Gamma _{Q_k } Q_k  - DQ_k  + w_k \},\label{eq:1}
\end{align}
where
\begin{align}
w_k  =& (e_{a_k }  \otimes e_{b_k }  \otimes e_{s_k } )r_k  + \gamma (e_{a_k }  \otimes e_{b_k }  \otimes e_{s_k } )(e_{s_{k'} } )^T \Pi _{\Gamma _{Q_k } Q_k } \Gamma _{Q_k } Q_k\nonumber\\
&- (e_{a_k }  \otimes e_{b_k }  \otimes e_{s_k } )(e_{a_k }  \otimes e_{b_k }  \otimes e_{s_k } )^T Q_k  - (DR + \gamma DP\Pi _{\Gamma _{Q_k } Q_k } \Gamma _{Q_k } Q_k  - DQ_k ),\label{eq:w}\\
=& (e_{a_k }  \otimes e_{b_k }  \otimes e_{s_k } )\delta _k  - (DR + \gamma DP\Pi _{\Gamma _{Q_k } Q_k } \Gamma _{Q_k } Q_k  - DQ_k ),\nonumber
\end{align}
is the stochastic noise, where all randomness in~\cref{algo:standard-Q-learning2} is encoded into a single vector,
$(s_k,a_k,b_k,r_k,s_k')$ is the sample in the $k$-th time-step, and $\delta _k : = r_k  + \gamma (e_{s_{k'} } )^T \Pi _{\Gamma _{Q_k } Q_k } \Gamma _{Q_k } Q_k  - (e_{a_k }  \otimes e_{b_k }  \otimes e_{s_k } )^T Q_k$ is called the TD-error.

Moreover, by definition, the noise term has a zero mean conditioned on $Q_k$, i.e., ${\mathbb E}[w_k|Q_k]=0$.
Invoking the optimal Bellman equation $(\gamma DP\Pi _{\Gamma _{Q^* } Q^* } \Gamma _{Q^* }  - D)Q^*  + DR = 0$ in~\eqref{eq:optimal-Q-Bellman-eq2}, \eqref{eq:1} can be further rewritten by
\begin{align}
(Q_{k + 1}  - Q^* ) = \underbrace {\{ I + \alpha (\gamma DP\Pi _{\Gamma _{Q_k } Q_k } \Gamma _{Q_k }  - D)\} }_{: = A_{Q_k } }(Q_k  - Q^* ) + \underbrace {\alpha \gamma DP(\Pi _{\Gamma _{Q_k } Q_k } \Gamma _{Q_k }  - \Pi _{\Gamma _{Q^* } Q^* } \Gamma _{Q^* } )Q^* }_{: = b_{Q_k } } + \alpha w_k .\label{eq:Q-learning-stochastic-recursion-form}
\end{align}
which is a linear switching system with an extra affine term, $b_{Q_k}:=\alpha \gamma DP(\Pi _{\Gamma _{Q_k } Q_k } \Gamma _{Q_k }  - \Pi _{\Gamma _{Q^* } Q^* } \Gamma _{Q^* } )Q^*$, and stochastic noise $\alpha w_k$.
Using the notation, the minimax Q-learning iteration can be concisely represented as the \emph{stochastic affine switching system}
\begin{align}
Q_{k + 1}  - Q^*  = A_{Q_k } (Q_k  - Q^* ) + b_{Q_k }  + \alpha w_k,\label{eq:swithcing-system-form}
\end{align}

Therefore, the convergence of minimax Q-learning can be reduced to analyzing the stability of the above affine switching system.
However, proving its stability poses a significant challenge due to the presence of affine and stochastic terms.
In the absence of these terms, we can establish the exponential stability of the corresponding deterministic switching system, irrespective of the switching policy (as demonstrated in~\cref{prop:stability}). However,~\eqref{eq:swithcing-system-form} includes additional affine terms and stochastic noises, making it unclear how to derive its finite-time convergence directly. To address this issue, we employ two simpler comparison systems that bound the trajectories of the original system and are more amenable to analysis.
The construction of these comparison systems draws inspiration from~\cite{lee2020unified} and~\cite{lee2021discrete}, capitalizing on the unique structure of the Q-learning algorithm.
Unlike previous works in~\cite{lee2020unified}, our focus lies in the discrete-time domain and finite-time analysis.
Moreover,~\cite{lee2021discrete} presents a finite-time analysis of standard Q-learning through the framework of discrete-time switching system models.
In contrast to~\cite{lee2021discrete}, the discrete-time switching system model in this paper exhibits a distinct structure including the min operator in its updates.
Consequently, establishing finite-time convergence becomes considerably more challenging. In other words, it is not feasible to adopt analogous approaches as those outlined in~\cite{lee2021discrete}, which will be elaborated in the remaining parts of this paper.
In the following subsections, we will present the two comparison systems, called the upper and lower comparison systems.

\subsection{Lower comparison system}
Let us consider the following stochastic switching system:
\begin{align}
(Q_{k + 1}^L  - Q^* ) = (I + \alpha \{ \gamma DP\Pi _{\Gamma _{Q^* } Q^* } \Gamma _{Q_k^L  - Q^* }  - D\} )(Q_k^L  - Q^* ) + \alpha w_{k},\quad Q_0^L-Q^*\in {\mathbb R}^{|{\cal S}\times {\cal A}\times {\cal B}|}, \label{eq:lower-system}
\end{align}
where the stochastic noise $w_{k}$ in the lower comparison system is the same as that in the original system~\eqref{eq:Q-learning-stochastic-recursion-form}. We refer to this system as the \emph{lower comparison system}. It is important to note that, unlike the original system~\eqref{eq:swithcing-system-form}, the above system does not include the affine term. Furthermore, its main property is that if $Q_0^L - Q^* \le Q_0-Q^*$ initially, then $Q_k^L - Q^* \le Q_k-Q^*$ holds for all $k \ge 0$.
\begin{proposition}\label{prop:lower-bound2}
Suppose $Q_0^L- Q^*\le Q_0 - Q^*$, where $\le$ is used as the element-wise inequality. Then,
$$Q_k^L- Q^*\le Q_k-Q^*,$$
for all $k\geq0$.
\end{proposition}
\begin{proof}
The proof is done by an induction argument. Suppose the result holds for some $k \ge 0$. Then,
\begin{align*}
(Q_{k+1}- Q^*)=& (Q_k  - Q^* ) + \alpha D\{ \gamma P\Pi _{\Gamma _{Q_k } Q_k } \Gamma _{Q_k } Q_k  - \gamma P\Pi _{\Gamma _{Q^* } Q^* } \Gamma _{Q^* } Q^*  - Q_k  + Q^* \}  + \alpha w_{k}\\
\ge& (Q_k  - Q^* ) + \alpha D\{ \gamma P\Pi _{\Gamma _{Q_k } Q_k } \Gamma _{Q_k } Q_k  - \gamma P\Pi _{\Gamma _{Q^* } Q^* } \Gamma _{Q_k } Q^*  - Q_k  + Q^* \}  + \alpha w_{k}\\
\ge& (Q_k  - Q^* ) + \alpha D\{ \gamma P\Pi _{\Gamma _{Q^* } Q^* } \Gamma _{Q_k } Q_k  - \gamma P\Pi _{\Gamma _{Q^* } Q^* } \Gamma _{Q_k } Q^*  - Q_k  + Q^* \}  + \alpha w_{k}\\
=& (I + \alpha \{ \gamma DP\Pi _{\Gamma _{Q^* } Q^* } \Gamma _{Q_k }  - D\} )(Q_k  - Q^* ) + \alpha w_{k}\\
\ge& (I + \alpha \{ \gamma DP\Pi _{\Gamma _{Q^* } Q^* } \Gamma _{Q_k }  - D\} )(Q_k^L  - Q^* ) + \alpha w_{k}\\
\ge& (I + \alpha \{ \gamma DP\Pi _{\Gamma _{Q^* } Q^* } \Gamma _{Q_k^L  - Q^* }  - D\} )(Q_k^L  - Q^* ) + \alpha w_{k}\\
=& Q_{k+1}^L-Q^*,
\end{align*}
where the third inequality is due to the hypothesis $Q_k^L- Q^*\le Q_k-Q^*$ and the fact that $A_{Q^*}$ is a positive matrix (all elements are nonnegative). The proof is completed by induction.
\end{proof}

\subsection{Upper comparison system}
Now, let us introduce the stochastic linear switching system
\begin{align}
(Q_{k + 1}^U  - Q^* ) = (I + \alpha \{ \gamma DP\Pi _{\Gamma _{Q^* } (Q_k^U  - Q^* )} \Gamma _{Q^* }  - D\} )(Q_k^U  - Q^* ) + \alpha w_{k},\quad Q_0^U-Q^*\in {\mathbb R}^{|{\cal S}\times {\cal A}\times {\cal B}|},
\label{eq:upper-system}
\end{align}
 where the stochastic noise $w_k$ is kept the same as the original system. We will call it the \emph{upper comparison system}.
 Similar to the lower comparison system~\eqref{eq:lower-system}, the above system does not include the affine term.
 Moreover, if $Q_0^U-Q^*\ge Q_0-Q^*$ initially, then $Q_k^U-Q^*\ge Q_k-Q^*$ for all $k \ge 0$.
\begin{proposition}
Suppose $Q_0^U-Q^*\ge Q_0-Q^*$, where $\geq $ is used as the element-wise inequality. Then,
$$Q_k^U-Q^*\ge Q_k-Q^*,$$
for all $k \ge 0$.
\end{proposition}
\begin{proof}
Suppose the result holds for some $k \ge 0$. Then,
\begin{align*}
Q_{k + 1}  - Q^*  =& (Q_k  - Q^* ) + \alpha D\{ \gamma P\Pi _{\Gamma _{Q_k } Q_k } \Gamma _{Q_k } Q_k  - \gamma P\Pi _{\Gamma _{Q^* } Q^* } \Gamma _{Q^* } Q^*  - Q_k  + Q^* \}  + \alpha w_k\\
\le& (Q_k  - Q^* ) + \alpha D\{ \gamma P\Pi _{\Gamma _{Q^* } Q_k } \Gamma _{Q^* } Q_k  - \gamma P\Pi _{\Gamma _{Q^* } Q^* } \Gamma _{Q^* } Q^*  - Q_k  + Q^* \}  + \alpha w_k\\
\le& (Q_k  - Q^* ) + \alpha D\{ \gamma P\Pi _{\Gamma _{Q^* } Q_k } \Gamma _{Q^* } Q_k  - \gamma P\Pi _{\Gamma _{Q^* } Q_k } \Gamma _{Q^* } Q^*  - Q_k  + Q^* \}  + \alpha w_k\\
\le& (I + \alpha \{ \gamma DP\Pi _{\Gamma _{Q^* } Q_k } \Gamma _{Q^* }  - D\} )(Q_k  - Q^* ) + \alpha w_k\\
\le& (I + \alpha \{ \gamma DP\Pi _{\Gamma _{Q^* } Q_k } \Gamma _{Q^* }  - D\} )(Q_k^U  - Q^* ) + \alpha w_k\\
\le& (I + \alpha \{ \gamma DP\Pi _{\Gamma _{Q^* } (Q_k^U  - Q^* )} \Gamma _{Q^* }  - D\} )(Q_k^U  - Q^* ) + \alpha w_k\\
=& Q_{k+1}^U-Q^*,
\end{align*}
where the third inequality is due to the hypothesis $Q_k^U-Q^*\ge Q_k-Q^*$ and the fact that $A_{Q_k}$ is a positive matrix. The proof is completed by induction.
\end{proof}

\section{Finite-time analysis of minimax Q-learning}
Building upon the constructions of the lower, upper, and original switching systems, we now proceed to establish the convergence of minimax Q-learning. Considering that the original system~\eqref{eq:swithcing-system-form} associated with~\cref{algo:standard-Q-learning2} is confined within the bounds set by the upper and lower comparison systems, respectively, its convergence can be proved by ensuring the convergence of these comparison systems. To begin, our attention is directed towards the lower comparison system in~\eqref{eq:lower-system}.

\subsection{Lower comparison system}
It is worth noting that the lower comparison system~\eqref{eq:lower-system}, being a switching system with stochastic noises, still presents challenges in establishing its convergence due to the complex dependence of the system matrices on the state.
To address this difficulty, we propose an additional pair of lower and upper comparison systems that effectively bound the lower comparison system described in~\eqref{eq:lower-system}. First and foremost, we will focus on the upper comparison system, called the {\em lower-upper comparison system}, which provides an upper bound for the lower comparison system. The formulation of this system is as follows:
\begin{align}
(Q_{k + 1}^{LU}  - Q^* ) = \underbrace {(I + \alpha \{ \gamma DP\Pi _{\Gamma _{Q^* } Q^* } \Gamma _{Q^* }  - D\} )}_{: = A}(Q_k^{LU}  - Q^* ) + \alpha w_k,\label{eq:lower-upper-system}
\end{align}
which is a {\em stochastic linear system}. We will prove that the trajectory of this lower-upper comparison system upper bounds that of the lower comparison system in~\eqref{eq:lower-system}.
\begin{proposition}\label{prop:lower-bound2}
Suppose $Q_0^L- Q^*\le Q_0^{LU} - Q^*$, where $\le$ is used as the element-wise inequality. Then,
$$Q_k^L- Q^*\le Q_k^{LU}-Q^*,$$
for all $k\geq0$.
\end{proposition}
\begin{proof}
The proof is done by an induction argument. Suppose the result holds for some $k \ge 0$. Then,
\begin{align*}
(Q_{k + 1}^L  - Q^* ) =& (I + \alpha \{ \gamma DP\Pi _{\Gamma _{Q^* } Q^* } \Gamma _{Q_k^L  - Q^* }  - D\} )(Q_k^L  - Q^* ) + \alpha w_{k}\\
\le& (I + \alpha \{ \gamma DP\Pi _{\Gamma _{Q^* } Q^* } \Gamma _{Q^* }  - D\} )(Q_k^L  - Q^* ) + \alpha w_k\\
\le& (I + \alpha \{ \gamma DP\Pi _{\Gamma _{Q^* } Q^* } \Gamma _{Q^* }  - D\} )(Q_k^{LU}  - Q^* ) + \alpha w_k\\
=& Q_{k+1}^{LU}-Q^*,
\end{align*}
where the second inequality is due to the hypothesis $Q_k^L- Q^*\le Q_k^{LU}-Q^*$ and the fact that $A_{Q^*}$ is a positive matrix (all elements are nonnegative). The proof is completed by induction.
\end{proof}

Defining $x_k := Q_k^{LU}  - Q^*$ and $A:= (I + \alpha \{ \gamma DP\Pi _{\Gamma _{Q^* } Q^* } \Gamma _{Q^* }  - D\} )$,~\eqref{eq:lower-upper-system} can be concisely represented as the stochastic linear system
\begin{align}
x_{k + 1} = A x_k + \alpha w_k,\quad x_0 \in {\mathbb R}^n,\quad \forall k \geq 0,\label{eq:linear-system-form}
\end{align}
where $n:= |{\cal S} \times {\cal A}\times {\cal B}|$, and $w_k\in {\mathbb R}^n$ is a stochastic noise.
We will first prove the convergence of this linear system, whose proof is given in Appendix~\ref{appendix:2}.
\begin{theorem}\label{thm:main1}
For any $k \geq 0$, we have
\begin{align}
{\mathbb E}[\| Q_k^{LU}  - Q^*\|_2 ] \le \frac{{3\alpha ^{1/2} |{\cal S} \times {\cal A}\times {\cal B}|}}{{d_{\min }^{1/2} (1 - \gamma )^{3/2} }} + |{\cal S} \times {\cal A}\times {\cal B}| \|Q_0^{LU}  - Q^*\|_2 \rho ^k.
\label{eq:8}
\end{align}
\end{theorem}

The first term on the right-hand side of~\eqref{eq:8} can be made arbitrarily small by reducing the step-size $\alpha \in (0,1)$. The second bound exponentially vanishes as $k \to \infty$ at the rate of $\rho = 1 - \alpha d_{\min} (1 - \gamma) \in (0,1)$. Therefore, it proves the exponential convergence of the mean-squared error of the lower comparison system up to a constant bias.

Although the convergence of the lower-upper comparison system has been established, it only guarantees convergence of an upper bound of the lower comparison system in~\eqref{eq:lower-system}. As mentioned earlier, it is hard to directly prove convergence of~\eqref{eq:lower-system} due to the dependance of the system matrix on $Q_k^L$.
In particular, if we take the expectation on both sides of~\eqref{eq:lower-system}, it is not possible to separate system matrix and the state unlike the lower-upper comparison system, making it much harder to analyze the stability of the lower comparison system.

To circumvent such a difficulty, we instead study an error system~\cite{lee2021discrete} by subtracting the lower comparison system from the lower-upper comparison system
\begin{align}
Q_{k + 1}^L  - Q_{k + 1}^{LU}  =& \underbrace {(I + \alpha \{ \gamma DP\Pi _{\Gamma _{Q^* } Q^* } \Gamma _{Q_k^L  - Q^* }  - D\} )}_{: = A_{Q_k^L } }(Q_k^L  - Q_{k}^{LU} )\nonumber \\
& + \underbrace {\alpha \gamma DP\Pi _{\Gamma _{Q^* } Q^* } (\Gamma _{Q_k^L  - Q^* }  - \Gamma _{Q^* } )}_{: = B_{Q_k^L } }(Q_{k}^{LU}  - Q^* )\label{eq:error-system}
\end{align}

Here, the stochastic noise $\alpha w_k$ is canceled out in the error system. Matrices $(A_{Q_k^L}, B_{Q_k^L})$ switch according to the external signal $Q_k^L$, and $Q_k^{LU}-Q^*$ can be seen as an external disturbance.
The key insight is as follows: if we can prove the stability of the error system, i.e., $Q_k^L-Q_k^{LU}\to 0$ as $k\to\infty$, then since $Q_k^{LU} \to Q^*$ as $k \to \infty$, we have $Q_k^L \to Q^*$ as well.
Keeping this picture in mind, we can establish the following bound on the expected error ${\mathbb E}[\left\| {Q_k^L  - Q^* } \right\|_\infty  ]$.
\begin{theorem}[Convergence]\label{thm:main2}
For all $k \geq 0$, we have
\begin{align}
{\mathbb E}[\left\| {Q_k^L  - Q^* } \right\|_\infty  ] \le& \frac{{9 d_{\max } |{\cal S} \times {\cal A}\times {\cal B}|\alpha ^{1/2} }}{{d_{\min }^{3/2} (1 - \gamma )^{5/2} }} + \frac{{2|{\cal S} \times {\cal A}\times {\cal B}|^{3/2} }}{{1 - \gamma }}\rho ^k + \frac{{4\alpha \gamma d_{\max } |{\cal S} \times {\cal A}\times {\cal B}|^{2/3} }}{{1 - \gamma }}k\rho ^{k - 1}\label{eq:3}
\end{align}
\end{theorem}
\begin{proof}
Taking norm on the error system in~\eqref{eq:error-system}, we get
\begin{align*}
\| {Q_{k + 1}^L  - Q_{k + 1}^{LU} } \|_\infty   \le& \| A_{Q_k^L } \|_\infty  \| Q_k^L  - Q_k^{LU} \|_\infty + \| B_{Q_k^L } \|_\infty  \| Q_k^{LU}  - Q^*\|_\infty\\
 \le& \rho \|Q_k^L  - Q_k^{LU}\|_\infty   + 2\alpha \gamma d_{\max } \|Q_k^{LU}  - Q^*\|_\infty
\end{align*}
where the second inequality is due to~\cref{lemma:max-norm-system-matrix} and the definition of $B_{Q_k^L }$ in~\eqref{eq:error-system}.
Taking the expectation on both sides of the last inequality and combining the last inequality with that in~\cref{thm:main1} yield
\begin{align*}
{\mathbb E}[\| Q_{i + 1}^L  - Q_{i + 1}^{LU} \|_\infty  ] \le& \rho {\mathbb E}[ \|Q_i^L  - Q_i^{LU}\|_\infty] + 2\alpha \gamma d_{\max } \frac{{3|{\cal S} \times {\cal A} \times {\cal B}|\alpha ^{1/2} }}{{d_{\min }^{1/2} (1 - \gamma )^{3/2} }} + 2\alpha \gamma d_{\max } \|Q_0^{LU}  - Q^*\|_2 |{\cal S} \times {\cal A}\times {\cal B}|\rho ^i
\end{align*}
for all $i\geq 0$. Summing the inequality from $i=0$ to $i=k$ and letting $Q_0^{LU}  = Q_0^L  = Q_0$ lead to
\begin{align}
{\mathbb E}[\| Q_k^L  - Q_k^{LU}\|_\infty  ] \le \frac{{6\gamma d_{\max } |{\cal S} \times {\cal A}\times {\cal B}|\alpha ^{1/2} }}{{d_{\min }^{3/2} (1 - \gamma )^{5/2} }} + k\rho ^{k - 1} 2\alpha \gamma d_{\max } \|Q_0  - Q^*\|_2 |{\cal S} \times {\cal A}\times {\cal B}|.\label{eq:10}
\end{align}
Using $\|Q_0 - Q^*\|_2  \le |{\cal S} \times {\cal A}\times {\cal B}|^{1/2} \|Q_0 - Q^*\|_\infty   \le |{\cal S} \times {\cal A}\times {\cal B}|^{1/2} \frac{2}{{1 - \gamma }}$ further leads to
\begin{align}
{\mathbb E}[\| Q_k^{L}  - Q_k^{LU} \|_\infty  ] \le& \frac{{6\gamma d_{\max } |{\cal S} \times {\cal A}\times {\cal B}|\alpha ^{1/2} }}{{d_{\min }^{3/2} (1 - \gamma )^{5/2} }} + k\rho ^{k - 1} 4 \alpha \gamma d_{\max } \frac{{|{\cal S} \times {\cal A}\times {\cal B}|^{2/3} }}{{1 - \gamma }}\label{eq:2}
\end{align}

On the other hand, using the triangle inequality leads to
\begin{align}
{\mathbb E}[\left\| {Q_k^L  - Q^* } \right\|_\infty  ] = {\mathbb E}[\left\| {Q_k^L  - Q_k^{LU}  + Q_k^{LU}  - Q^* } \right\|_\infty  ] \le {\mathbb E}[\left\| {Q_k^{LU}  - Q^* } \right\|_\infty  ] + {\mathbb E}[\left\| {Q_k^L  - Q_k^{LU} } \right\|_\infty  ]\label{eq:11}
\end{align}
Combining~\eqref{eq:11} with that in~\cref{thm:main1} leads to
\begin{align*}
{\mathbb E}[\| Q_k^L  - Q^*\|_\infty  ] \le& {\mathbb E}[ \|Q_k^{LU}  - Q^*\|_\infty  ] + {\mathbb E}[\|Q_k^L  - Q_k^{LU}\|_\infty  ]\\
\le& \frac{{3\alpha ^{1/2} |{\cal S} \times {\cal A}\times {\cal B}|}}{{d_{\min }^{1/2} (1 - \gamma )^{3/2} }} + \frac{{2|{\cal S} \times {\cal A}\times {\cal B}|^{3/2} }}{{1 - \gamma }}\rho ^k  + {\mathbb E}[\|Q_k^L  - Q_k^{LU}\|_\infty  ].
\end{align*}
Moreover, combining the above inequality with~\eqref{eq:2} yields the desired conclusion.

\end{proof}

Note that the first term in~\eqref{eq:3} is the constant error due to the constant step-size, which is scaled according to $\alpha \in (0,1)$. The second term in~\eqref{eq:3} is due to the gap between lower comparison system and original system, and the third term in~\eqref{eq:3} is due to the gap between upper comparison system and original system. The second term ${\cal O}(\rho^k)$ exponentially decays, and the third term ${\cal O}(k \rho^{k-1})$ also exponentially decays while the speed is slower than the second term due to the additional linearly increasing factor. The upper bound in~\eqref{eq:3} can be converted to looser but more interpretable forms as follows.
\begin{corollary}\label{thm:main3}
For any $k \geq 0$, we have
\begin{align}
{\mathbb E}[\|Q_k^L  - Q^*\|_\infty  ] \le& \frac{{9 d_{\max } |{\cal S} \times {\cal A}\times {\cal B}|\alpha ^{1/2} }}{{d_{\min }^{3/2} (1 - \gamma )^{5/2} }} + \frac{{2|{\cal S} \times {\cal A}\times {\cal B}|^{3/2} }}{{1 - \gamma }}\rho ^k \nonumber \\
& + \frac{{4\alpha \gamma d_{\max } |{\cal S} \times {\cal A}\times {\cal B}|^{2/3} }}{{1 - \gamma }}\frac{{ - 2}}{{\ln (\rho )}}\rho ^{\frac{{ - 1}}{{\ln (\rho )}} - 1} \rho ^{k/2}\label{eq:4}
\end{align}
and
\begin{align}
{\mathbb E}[\left\| {Q_k^L  - Q^* } \right\|_\infty  ] \le& \frac{{9 d_{\max } |{\cal S} \times {\cal A}\times {\cal B}|\alpha ^{1/2} }}{{d_{\min }^{3/2} (1 - \gamma )^{5/2} }} + \frac{{2|{\cal S} \times {\cal A}\times {\cal B}|^{3/2} }}{{1 - \gamma }}\rho ^k\nonumber\\
& + \frac{{8\gamma d_{\max } |{\cal S} \times {\cal A}\times {\cal B}|^{2/3} }}{{1 - \gamma }}\frac{1}{{d_{\min } (1 - \gamma )}}\rho ^{k/2 - 1}\label{eq:5}
\end{align}

\end{corollary}
\begin{proof}
In~\eqref{eq:3}, we focus on the term
\[
k\rho ^{k - 1}  = k\rho ^{k/2 + k/2 - 1}  = k\rho ^{k/2 - 1} \rho ^{k/2}
\]

Let $f(x) = x\rho ^{x/2}  = x\rho ^{x/2}$. Checking the first-order optimality condition
\[
\frac{{df(x)}}{{dx}} = \frac{d}{{dx}}x\rho ^{x/2}  = \rho ^{x/2}  + x\frac{1}{2}\rho ^{x/2} \ln (\rho ) = 0
\]
it follows that its maximum point is $x = \frac{{ - 2}}{{\ln (\rho )}}$, and the corresponding maximum value is
\[
f\left( {\frac{{ - 2}}{{\ln (\rho )}}} \right) = \frac{{ - 2}}{{\ln (\rho )}}\rho ^{\frac{{ - 1}}{{\ln (\rho )}}}
\]

Therefore, we have the bounds
\[
k\rho ^{k - 1}  = k\rho ^{k/2} \rho ^{ - 1} \rho ^{k/2}  \le \frac{{ - 2}}{{\ln (\rho )}}\rho ^{\frac{{ - 1}}{{\ln (\rho )}} - 1} \rho ^{k/2}
\]

Combining this bound with~\eqref{eq:3}, one gets the first bound in~\eqref{eq:4}.
To obtain the second inequality in~\eqref{eq:5}, we use the relation $1 - \frac{1}{x} \le \ln x \le x - 1,\forall x > 0$ to obtain
\begin{align*}
\frac{1}{{\ln (\rho ^{ - 1} )}}\rho ^{\frac{1}{{\ln (\rho ^{ - 1} )}}}  \le& \frac{1}{{1 - \frac{1}{{\rho ^{ - 1} }}}}\rho ^{\rho ^{ - 1}  - 1}\\
\le& \frac{1}{{\alpha d_{\min } (1 - \gamma )}}\rho ^{\frac{1}{{1 - \alpha d_{\min } (1 - \gamma )}} - 1}\\
\le& \frac{1}{{\alpha d_{\min } (1 - \gamma )}}
\end{align*}
where the last inequality uses $\alpha \in (0,1)$ in~\cref{assumption:step-size}.
Combining the above bound with~\eqref{eq:4},~\eqref{eq:5} follows. This completes the proof.
\end{proof}

\subsection{Upper comparison system}
Until now, we have established a finite-time analysis for the lower comparison system. In a similar manner, one can also offer a finite-time analysis for the upper comparison system. Due to the symmetrical nature of the upper comparison system in relation to the lower comparison system, we will omit the comprehensive derivation process for the sake of brevity in this presentation.
In particular, for the upper comparison system~\eqref{eq:upper-system}, let us define the following system:
\begin{align}
(Q_{k + 1}^{UL}  - Q^* ) = \underbrace {(I + \alpha \{ \gamma DP\Pi _{\Gamma _{Q^* } Q^* } \Gamma _{Q^* }  - D\} )}_{: = A}(Q_k^{UL}  - Q^* ) + \alpha w_k ,\label{eq:upper-lower-system}
\end{align}
which is a stochastic linear system, and lower bounds the upper comparison system. For this reason, let us call it a {\em upper-lower comparison system}. We can prove that the trajectory of this upper-lower comparison system upper bounds that of the upper comparison system in~\eqref{eq:upper-system}.
\begin{proposition}\label{prop:upper-bound2}
Suppose $Q_0^U- Q^* \geq Q_0^{UL} - Q^*$, where $\geq$ is used as the element-wise inequality. Then,
$$Q_k^U- Q^*\geq Q_k^{UL}-Q^*,$$
for all $k\geq0$.
\end{proposition}
\begin{proof}
The proof is done by an induction argument. Suppose the result holds for some $k \ge 0$. Then,
\begin{align*}
(Q_{k + 1}^U  - Q^* ) =& (I + \alpha \{ \gamma DP\Pi _{\Gamma _{Q^* } (Q_k^U  - Q^* )} \Gamma _{Q^* }  - D\} )(Q_k^U  - Q^* ) + \alpha w_k\\
\ge& (I + \alpha \{ \gamma DP\Pi _{\Gamma _{Q^* } Q^* } \Gamma _{Q^* }  - D\} )(Q_k^U  - Q^* ) + \alpha w_k\\
\ge& (I + \alpha \{ \gamma DP\Pi _{\Gamma _{Q^* } Q^* } \Gamma _{Q^* }  - D\} )(Q_k^{UL}  - Q^* ) + \alpha w_k\\
=& Q_{k + 1}^{UL}  - Q^*
\end{align*}
where the second inequality is due to the hypothesis $Q_k^U- Q^*\geq  Q_k^{UL}-Q^*$ and the fact that $A_{Q^*}$ is a positive matrix (all elements are nonnegative). The proof is completed by induction.
\end{proof}

Defining $x_k := Q_k^{UL}  - Q^*$ and $A:= (I + \alpha \{ \gamma DP\Pi _{\Gamma _{Q^* } Q^* } \Gamma _{Q^* }  - D\} )$,~\eqref{eq:upper-lower-system} can be represented as the stochastic linear system~\eqref{eq:linear-system-form}.
Its finite-time bound is identical to that in~\cref{thm:main1}. Moreover, subtracting the upper comparison system from the upper-lower comparison system leads to
\begin{align}
Q_{k + 1}^U  - Q_{k + 1}^{UL}  =& \underbrace {\left( {I + \alpha \{ \gamma DP\Pi _{\Gamma _{Q^* } (Q_k^U  - Q^* )} \Gamma _{Q^* }  - D\} } \right)}_{: = A_{Q_k^U } }(Q_k^U  - Q_k^{UL} )\nonumber\\
& + \underbrace {\alpha \gamma DP(\Pi _{\Gamma _{Q^* } (Q_k^U  - Q^* )}  - \Pi _{\Gamma _{Q^* } Q^* } )\Gamma _{Q^* } }_{: = B_{Q_k^U } }(Q_k^{UL}  - Q^* )
\label{eq:error-system2}
\end{align}

Using this system and following similar lines as in the lower comparison system, we can establish a finite-time error bound on the upper system's error.
\begin{theorem}\label{thm:main4}
For any $k \geq 0$, we have
\begin{align}
{\mathbb E}[\left\| {Q_k^U  - Q^* } \right\|_\infty  ] \le& \frac{{9 d_{\max } |{\cal S} \times {\cal A}\times {\cal B}|\alpha ^{1/2} }}{{d_{\min }^{3/2} (1 - \gamma )^{5/2} }} + \frac{{2|{\cal S} \times {\cal A}\times {\cal B}|^{3/2} }}{{1 - \gamma }}\rho ^k\nonumber\\
& + \frac{{8\gamma d_{\max } |{\cal S} \times {\cal A}\times {\cal B}|^{2/3} }}{{1 - \gamma }}\frac{1}{{d_{\min } (1 - \gamma )}}\rho ^{k/2 - 1}\label{eq:7}
\end{align}
\end{theorem}
\begin{proof}
Taking norm on the error system in~\eqref{eq:error-system2}, we get
\begin{align*}
\| {Q_{k + 1}^U  - Q_{k + 1}^{UL} } \|_\infty   \le& \| A_{Q_k^U } \|_\infty  \| Q_k^U  - Q_k^{UL} \|_\infty + \| B_{Q_k^U } \|_\infty  \| Q_k^{UL}  - Q^*\|_\infty\\
 \le& \rho \|Q_k^U  - Q_k^{UL}\|_\infty   + 2\alpha \gamma d_{\max } \|Q_k^{UL}  - Q^*\|_\infty
\end{align*}
where the second inequality is due to~\cref{lemma:max-norm-system-matrix} and the definition of $B_{Q_k^U}$ in~\eqref{eq:error-system2}. The remaining parts of the proof are essentially identical to the proof of~\cref{thm:main2}, and hence, they are omitted here for brevity.
\end{proof}

Similarly to~\cref{thm:main3}, the upper bound in~\eqref{eq:7} can be converted to looser but more interpretable forms as follows.
\begin{corollary}\label{thm:main6}
For any $k \geq 0$, we have
\begin{align*}
{\mathbb E}[\|Q_k^U - Q^*\|_\infty  ] \le& \frac{{9 d_{\max } |{\cal S} \times {\cal A}\times {\cal B}|\alpha ^{1/2} }}{{d_{\min }^{3/2} (1 - \gamma )^{5/2} }} + \frac{{2|{\cal S} \times {\cal A}\times {\cal B}|^{3/2} }}{{1 - \gamma }}\rho ^k \nonumber \\
& + \frac{{4\alpha \gamma d_{\max } |{\cal S} \times {\cal A}\times {\cal B}|^{2/3} }}{{1 - \gamma }}\frac{{ - 2}}{{\ln (\rho )}}\rho ^{\frac{{ - 1}}{{\ln (\rho )}} - 1} \rho ^{k/2}
\end{align*}
and
\begin{align*}
{\mathbb E}[\left\| {Q_k^U - Q^* } \right\|_\infty  ] \le& \frac{{9 d_{\max } |{\cal S} \times {\cal A}\times {\cal B}|\alpha ^{1/2} }}{{d_{\min }^{3/2} (1 - \gamma )^{5/2} }} + \frac{{2|{\cal S} \times {\cal A}\times {\cal B}|^{3/2} }}{{1 - \gamma }}\rho ^k\nonumber\\
& + \frac{{8\gamma d_{\max } |{\cal S} \times {\cal A}\times {\cal B}|^{2/3} }}{{1 - \gamma }}\frac{1}{{d_{\min } (1 - \gamma )}}\rho ^{k/2 - 1}
\end{align*}
\end{corollary}

Now,~\cref{thm:main3} and~\cref{thm:main6} offer finite-time error bounds for the upper and lower comparison systems, respectively. By merging these bounds, we can deduce an upper bound on the original switching system~\eqref{eq:swithcing-system-form}.
\begin{theorem}\label{thm:main5}
For any $k \geq 0$, we have
\[
{\mathbb E}[\left\| {Q_k  - Q^* } \right\|_2 ] \le \frac{{27d_{\max } |{\cal S} \times {\cal A} \times {\cal B}|\alpha ^{1/2} }}{{d_{\min }^{3/2} (1 - \gamma )^{5/2} }} + \frac{{6|{\cal S} \times {\cal A} \times {\cal B}|^{3/2} }}{{1 - \gamma }}\rho ^k  + \frac{{24\gamma d_{\max } |{\cal S} \times {\cal A} \times {\cal B}|^{2/3} }}{{1 - \gamma }}\frac{3}{{d_{\min } (1 - \gamma )}}\rho ^{k/2 - 1}
\]
\end{theorem}
\begin{proof}
We have
\begin{align*}
{\mathbb E}[\left\| {Q_k  - Q^* } \right\|_2 ] =& {\mathbb E}[\left\| {Q_k  - Q_k^L  + Q_k^L  - Q^* } \right\|_2 ]\\
\le& {\mathbb E}[\left\| {Q_k^L  - Q^* } \right\|_2 ] + {\mathbb E}[\left\| {Q_k  - Q_k^L } \right\|_2 ]\\
 \le& {\mathbb E}[\left\| {Q_k^L  - Q^* } \right\|_2 ] + {\mathbb E}[\left\| {Q_k^U  - Q_k^L } \right\|_2 ]\\
 \le& {\mathbb E}[\left\| {Q_k^L  - Q^* } \right\|_2 ] + {\mathbb E}[\left\| {Q_k^U  - Q^*  + Q^*  - Q_k^L } \right\|_2 ]\\
 \le& {\mathbb E}[\left\| {Q_k^L  - Q^* } \right\|_2 ] + {\mathbb E}[\left\| {Q_k^U  - Q^* } \right\|_2 ] + {\mathbb E}[\left\| {Q^*  - Q_k^L } \right\|_2 ]\\
 =& 2{\mathbb E}[\left\| {Q_k^L  - Q^* } \right\|_2 ] + {\mathbb E}[\left\| {Q_k^U  - Q^* } \right\|_2 ]
\end{align*}
where the first and fourth inequalities come from the triangle inequality, and the second is due to the fact that $Q_k^U  - Q_k^L  \ge Q_k  - Q_k^L  \ge 0$. Combining the last inequality with~\cref{thm:main3} and~\cref{thm:main6}, one gets the desired conclusion.
\end{proof}

\section*{Conclusion}
This paper has investigated the finite-time analysis of the minimax Q-learning algorithm applied to two-player zero-sum Markov games. Additionally, we have established a finite-time analysis of the associated Q-value iteration method. To conduct our analysis, we employed switching system models for both minimax Q-learning and Q-value iteration. We anticipate that this approach provides deeper insights into minimax Q-learning and facilitates a more straightforward and insightful convergence analysis. Furthermore, these additional insights hold the potential to uncover new connections and foster collaboration between concepts in the domains of control theory and reinforcement learning communities.

\bibliographystyle{IEEEtran}
\bibliography{reference}

\appendices

\section{Technical lemmas}
Let us consider the stochastic linear system~\eqref{eq:linear-system-form}.
The noise $w_k$ has the zero mean conditioned on $Q_k$, and is bounded. These properties are formally proved in the following lemma.
\begin{lemma}\label{lemma:bound-W}
We have
\begin{enumerate}
\item ${\mathbb E}[w_k] = 0$;

\item ${\mathbb E}[\left\| {w_k } \right\|_\infty  ] \le \sqrt{W_{\max }}$;

\item ${\mathbb E}[\left\| {w_k } \right\|_2 ] \le \sqrt{W_{\max }}$;

\item ${\mathbb E}[w_k^T w_k ] \le \frac{9}{{(1 - \gamma )^2 }} = :W_{\max }$.
\end{enumerate}
for all $k\geq 0$.
\end{lemma}
\begin{proof}
For the first statement, we take the conditional expectation on~\eqref{eq:w} to have ${\mathbb E}[w_k |x_k ] = 0$. Taking the total expectation again with the law of total expectation leads to the first conclusion. Moreover, the conditional expectation, ${\mathbb E}[w_k^T w_k |Q_k ]$, is bounded as
\begin{align*}
{\mathbb E}[w_k^T w_k |Q_k ] =& {\mathbb E}[\left\| {w_k } \right\|_2^2 |Q_k ]\\
=&{\mathbb E}\left[ {\left. {\left\| {(e_{a_k }  \otimes e_{b_k }  \otimes e_{s_k } )\delta _k  - (DR + \gamma DP\Pi _{\Gamma _{Q_k } Q_k } \Gamma _{Q_k } Q_k  - DQ_k )} \right\|_2^2 } \right|Q_k } \right]\\
=& {\mathbb E}[\delta _k^2 |Q_k ] - \left\| {DR + \gamma DP\Pi _{\Gamma _{Q_k } Q_k } \Gamma _{Q_k } Q_k  - DQ_k } \right\|_2^2\\
\le& {\mathbb E}[\delta _k^2 |Q_k ]\\
=& {\mathbb E}[r_k^2 |Q_k ] + {\mathbb E}[2r_k \gamma (e_{s_k '} )^T \Pi _{\Gamma _{Q_k } Q_k } \Gamma _{Q_k } Q_k |Q_k ]\\
&+ {\mathbb E}[ - 2r_k (e_{a_k}  \otimes e_{b_k}  \otimes e_{s_k} )^T Q_k |Q_k ]\\
&+ {\mathbb E}[ - 2\gamma (e_{s_k '} )^T \Pi _{\Gamma _{Q_k } Q_k } \Gamma _{Q_k } Q_k (e_{a_k }  \otimes e_{b_k }  \otimes e_{s_k } )^T Q_k |Q_k ]\\
&+ {\mathbb E}[\gamma (e_{s_k '} )^T \Pi _{\Gamma _{Q_k } Q_k } \Gamma_{Q_k } Q_k \gamma (e_{s_k '} )^T \Pi _{\Gamma _{Q_k } Q_k } \Gamma _{Q_k } Q_k |Q_k ]\\
& + {\mathbb E}[(e_{a_k}  \otimes e_{s_k} )^T Q_k (e_{a_k}  \otimes e_{s_k} )^T Q_k |Q_k ]\\
\le& 1 + 2\gamma E[|r_k ||(e_{s_k '} )^T \Pi _{\Gamma _{Q_k } Q_k } \Gamma _{Q_k } Q_k ||Q_k ]\\
& + 2{\mathbb E}[|r_k ||(e_{a_k}  \otimes e_{b_k} \otimes e_{s_k} )^T Q_k ||Q_k ]\\
&+ 2\gamma {\mathbb E}[|(e_{s_k '} )^T \Pi _{\Gamma _{Q_k } Q_k } \Gamma _{Q_k } Q_k ||(e_{a_k }  \otimes e_{b_k }  \otimes e_{s_k } )^T Q_k ||Q_k ]\\
& + \gamma ^2 {\mathbb E}[|(e_{s_k '} )^T \Pi _{\Gamma _{Q_k } Q_k } \Gamma _{Q_k } Q_k ||(e_{s_k '} )^T \Pi _{\Gamma _{Q_k } Q_k } \Gamma _{Q_k } Q_k ||Q_k ]\\
&+ {\mathbb E}[|(e_{a_k}  \otimes e_{b_k}  \otimes e_{s_k} )^T Q_k ||(e_{a_k}  \otimes e_{b_k}  \otimes e_{s_k} )^T Q_k ||Q_k ]\\
\le& \frac{9}{{(1 - \gamma )^2 }}=:W_{\max},
\end{align*}
where $\delta_k$ is defined as $\delta _k : = r_k  + \gamma (e_{s_{k'} } )^T \Pi _{\Gamma _{Q_k } Q_k } \Gamma _{Q_k } Q_k  - (e_{a_k }  \otimes e_{b_k }  \otimes e_{s_k } )^T Q_k$, and the last inequality comes from Assumptions~\ref{assumption:bounded-reward}-\ref{assumption:bounded-Q0}, and~\cref{lemma:bounded-Q}. Taking the total expectation, we have the fourth result. Next, taking the square root on both sides of ${\mathbb E}[\left\| {w_k } \right\|_2^2 ] \le W_{\max}$, one gets
\begin{align*}
{\mathbb E}[\left\| {w_k } \right\|_\infty  ] \le {\mathbb E}[\left\| {w_k } \right\|_2 ] \le \sqrt {{\mathbb E}[\left\| {w_k } \right\|_2^2 ]}  \le \sqrt {W_{\max}},
\end{align*}
where the first inequality comes from $\left\|  \cdot  \right\|_\infty   \le \left\|  \cdot  \right\|_2$. This completes the proof.
\end{proof}

\section{Proof of~\cref{lemma:bounded-Q}}\label{appendix:1}

The update of minimax Q-learning in~\cref{algo:standard-Q-learning2} can be written compactly as
\begin{align}
Q_{k + 1}  = Q_k  + \alpha (e_{a_k }  \otimes e_{b_k }  \otimes e_{s_k } )(r_k  + \gamma (e_{s_k '} )^T \Pi _{\Gamma _{Q_k } Q_k } \Gamma _{Q_k } Q_k  - (e_{a_k }  \otimes e_{b_k }  \otimes e_{s_k } )^T Q_k )
\label{eq:14}
\end{align}

Taking the infinity norm of both sides of~\eqref{eq:14} with $k=0$ and using~\cref{assumption:bounded-reward}, we have
\begin{align*}
\left\| {Q_1 } \right\|_\infty   \le& (1 - \alpha )\left\| {Q_0 } \right\|_\infty   + \alpha (1 + \gamma \left\| {Q_0 } \right\|_\infty  )\\
\le& (1 - \alpha  + \alpha \gamma  + \alpha )\max \left\{ {1,\left\| {Q_0 } \right\|_\infty  } \right\}\\
\le& (1 + \gamma )\max \left\{ {1,\left\| {Q_0 } \right\|_\infty  } \right\}.
\end{align*}
For induction argument, assume $\left\| {Q_k } \right\|_\infty   \le (1 + \gamma  +  \cdots  + \gamma ^k )\max \left\{ {1,\left\| {Q_0 } \right\|_\infty  } \right\}$.
Then, taking the infinity norm of both sides of~\eqref{eq:14} leads to
\begin{align*}
\left\| {Q_{k + 1} } \right\|_\infty   \le& (1 - \alpha )\left\| {Q_k } \right\|_\infty   + \alpha (|r_k | + \gamma \left\| {Q_k } \right\|_\infty  )\\
 \le& (1 - \alpha )\left\| {Q_k } \right\|_\infty   + \alpha  + \gamma \alpha \left\| {Q_k } \right\|_\infty\\
\le& (1 - \alpha )(1 + \gamma  +  \cdots  + \gamma ^k )\max \left\{ {1,\left\| {Q_0 } \right\|_\infty  } \right\} + \alpha  + \gamma \alpha (1 + \gamma  +  \cdots  + \gamma ^k )\max \left\{ {1,\left\| {Q_0 } \right\|_\infty  } \right\}\\
\le& (1 - \alpha )(1 + \gamma  +  \cdots  + \gamma ^k )\max \left\{ {1,\left\| {Q_0 } \right\|_\infty  } \right\} + \alpha (1 + \gamma  +  \cdots  + \gamma ^{k + 1} )\max \left\{ {1,\left\| {Q_0 } \right\|_\infty  } \right\}\\
=& (1 + \gamma  +  \cdots  + \gamma ^k  + \gamma ^{k + 1} )\max \left\{ {1,\left\| {Q_0 } \right\|_\infty  } \right\}
\end{align*}
where the second inequality is due to the boundednes of rewards in~\cref{assumption:bounded-reward}, the third inequality follows from the hypothesis, and the fourth inequality is due to $\max \left\{ {1,\left\| {Q_0 } \right\|_\infty  } \right\} \ge 1$. By induction, we have
\[
\left\| {Q_k } \right\|_\infty   \le (1 + \gamma  +  \cdots  + \gamma ^k )\max \left\{ {1,\left\| {Q_0 } \right\|_\infty  } \right\} \le \frac{{\max \left\{ {1,\left\| {Q_0 } \right\|_\infty  } \right\}}}{{1 - \gamma }},\quad \forall k \ge 0,
\]
which completes the proof.

\section{Proof of~\cref{thm:main1}}\label{appendix:2}

Let us consider the stochastic linear system~\eqref{eq:linear-system-form}.
We first investigate how the autocorrelation matrix, ${\mathbb E}[x_k x_k^T ]$, propagates over the time. In particular, the autocorrelation matrix is updated through the linear recursion ${\mathbb E}[x_{k + 1} x_{k + 1}^T ] = A {\mathbb E}[x_k x_k^T ]A^T  + \alpha ^2 W_k$, where ${\mathbb E}[w_k w_k^T ] = W_k \succeq 0$ is the covariance of the noise. Defining $X_k := {\mathbb E}[x_k x_k^T ], k \geq 0$, it is equivalently written as the matrix recursion $X_{k + 1}  = AX_k A^T  + \alpha^2 W_k,\quad \forall k\geq 0$ with $X_0  := x_0x_0^T$. From this observation, one has $X_k  = \alpha ^2 \sum\limits_{i = 0}^{k - 1} {A^i W_{k - i - 1} (A^T )^i }  + A^k X_0 (A^T )^k$. Therefore, one can derive
\begin{align*}
{\mathbb E}[\|Q_k^{LU}  - Q^*\|_2^2 ] =& {\mathbb E}[(Q_k^{LU}  - Q^* )^T (Q_k^{LU}  - Q^* )]\\
 =& {\mathbb E}[{\rm tr}((Q_k^{LU}  - Q^* )^T (Q_k^{LU} - Q^* ))]\\
 =& {\mathbb E}[{\rm tr}((Q_k^{LU}  - Q^* )(Q_k^{LU} - Q^* )^T )]\\
 =& {\mathbb E}[{\rm tr}(X_k )]\\
\leq& n \lambda_{\max } (X_k )\\
\le& n \alpha ^2 \sum_{i = 0}^{k - 1} {\lambda_{\max } (A^i W_{k - i - 1} (A^T )^i )}  + n \lambda_{\max } (A^k X_0 (A^T )^k )\\
\le& n \alpha ^2 \sup _{j \ge 0} \lambda_{\max } (W_j )\sum_{i = 0}^{k - 1} {\lambda_{\max } (A^i (A^T )^i )}\\
&  + n \lambda_{\max } (X_0 )\lambda_{\max } (A^k (A^T )^k )\\
=& n \alpha ^2 \sup _{j \ge 0} \lambda_{\max } (W_j ) \sum_{i = 0}^{k - 1} {\| {A^i }\|_2^2 }  + n \lambda_{\max } (X_0 )\| {A^k } \|_2^2,
\end{align*}
where the first inequality comes from the fact that since $X_k \succeq 0$, the diagonal elements are nonnegative, and we have ${\rm{tr}}(X_k ) \le n\lambda_{\max } (X_k )$. Moreover, the second inequality is due to $A^i W_{k-i-1}(A^T )^i \succeq 0$ and $A^k X_0 (A^T )^k \succeq 0$. Next, the maximum eigenvalue of $W_k$ is bounded as $\lambda_{\max } (W_k) \le W_{\max}$ for all $k \geq 0$, where $W_{\max} > 0$ is given in~\cref{lemma:bound-W}. This is because $\lambda_{\max } (W_k) \le {\rm tr}(W_k) = {\rm tr}({\mathbb E}[w_k w_k^T ]) = {\mathbb E}[{\rm tr}(w_k w_k^T )] = {\mathbb E}[w_k^T w_k ] \le W_{\max}$, where the second inequality comes from~\cref{lemma:bound-W}, and the first equality uses the fact that the trace is a linear function.
Therefore, one gets
\begin{align*}
{\mathbb E}[\|Q_k^{LU}  - Q^*\|_2^2 ]\le& \alpha ^2  W_{\max } n^2 \sum_{i = 0}^{k - 1} {\| {A^i } \|_\infty ^2 }  + n^2\lambda_{\max } (X_0 )\| {A^k } \|_\infty ^2\\
\le& \alpha ^2 W_{\max } n^2 \sum_{i = 0}^{k - 1} {\rho ^{2i} }  + n^2 \lambda_{\max } (X_0 )\rho ^{2k}\\
\le& \alpha ^2 W_{\max } n^2 \lim_{k \to \infty } \sum_{i = 0}^{k - 1} {\rho ^{2i} }  + n^2 \lambda_{\max } (X_0 )\rho ^{2k}\\
\le& \frac{{\alpha ^2 W_{\max } n^2}}{{1 - \rho ^2 }} + n^2 \lambda_{\max } (X_0 )\rho ^{2k}\\
\le& \frac{{\alpha ^2 W_{\max } n^2}}{{1 - \rho }} + n^2 \lambda_{\max } (X_0 )\rho ^{2k}\\
\le& \frac{{\alpha W_{\max } n^2 }}{{d_{\min } (1 - \gamma )}} + n^2 \left\| {x_0 } \right\|_2^2 \rho ^{2k},
\end{align*}
where the first inequality is due to $\left\|  \cdot  \right\|_2  \le \sqrt n \left\|  \cdot  \right\|_\infty$, the second inequality is due to~\cref{lemma:max-norm-system-matrix}, the third and fourth inequalities come from $\rho \in (0,1)$, the last inequality is due to $\lambda_{\max } (X_0) \le {\rm tr}(X_0 ) = {\rm tr}(x_0 x_0^T ) = \left\| {x_0 } \right\|_2^2$ and $\rho = 1-\alpha d_{\min}(1-\gamma)$. Taking the square root on both side of the last inequality, using the subadditivity of the square root function, the Jensen inequality, and the concavity of the square root function, we have the desired conclusion.

\end{document}